\documentclass[pdfa,a4paper,UKenglish,cleveref,autoref,thm-restate]{lipics-v2021}
\usepackage{subfiles} 

\usepackage{amsmath}
\usepackage{amsthm}
\usepackage{float}
\usepackage{mathtools}
\usepackage{float}
\usepackage{tcolorbox}
\usepackage{multicol}
\usepackage[noend]{algpseudocode}
\usepackage{thmtools}
\usepackage{tcolorbox}
\usepackage{algorithm}
\newcommand{\E}{\mathbb{E}}
\newcommand{\eps}{\varepsilon}
 
\theoremstyle{definition}

\title{Competitive Capacitated Online Recoloring}

\author{Rajmohan {Rajaraman}}{Northeastern University, Boston, MA, USA \and \url{https://www.khoury.northeastern.edu/home/rraj/}}{r.rajaraman@northeastern.edu}{}{}
\author{Omer Wasim}{Northeastern University, Boston, MA, USA \and \url{https://sites.google.com/view/omer-wasim/home}}{wasim.o@northeastern.edu}{}{}

\authorrunning{R. Rajaraman and O. Wasim}
\Copyright{Rajmohan {Rajaraman} and Omer Wasim}

\ccsdesc{Theory of computation~Online algorithms}

\keywords{online algorithms, competitive ratio, recoloring, resource augmentation}

\acknowledgements{This work was partially supported by NSF grant CCF-2335187.}

\EventEditors{Timothy Chan, Johannes Fischer, John Iacono, and Grzegorz Herman}
\EventNoEds{4}
\EventLongTitle{32nd Annual European Symposium on Algorithms (ESA 2024)}
\EventShortTitle{ESA 2024}
\EventAcronym{ESA}
\EventYear{2024}
\EventDate{September 2--4, 2024}
\EventLocation{Royal Holloway, London, United Kingdom}
\EventLogo{}
\SeriesVolume{308}
\ArticleNo{96}

\nolinenumbers

\begin{document}

\maketitle

\begin{abstract}
In this paper, we revisit the \emph{online recoloring} problem introduced recently by Azar, Machluf, Patt-Shamir and Touitou \cite{azar_et_al:LIPIcs.ICALP.2022.13} to investigate algorithmic challenges that arise while scheduling virtual machines or processes in distributed systems and cloud services. In online recoloring, there is a fixed set $V$ of $n$ vertices and an initial coloring $c_0: V\rightarrow [k]$ for some $k\in \mathbb{Z}^{>0}$. Under an online sequence $\sigma$ of requests where each request is an edge $(u_t,v_t)$, a proper vertex coloring $c$ of the graph $G_t$ induced by requests until time $t$ needs to be maintained for all $t$; i.e., for any $(u,v)\in G_t$, $c(u)\neq c(v)$.  In the distributed systems application, a vertex corresponds to a VM, an edge corresponds to the requirement that the two endpoint VMs be on different clusters, and a coloring is an allocation of VMs to clusters.  The objective is to minimize the total weight of vertices recolored for the sequence $\sigma$. In \cite{azar_et_al:LIPIcs.ICALP.2022.13}, the authors give competitive algorithms for two polynomially tractable cases -- $2$-coloring for bipartite $G_t$ and $(\Delta+1)$-coloring for $\Delta$-degree $G_t$ -- and lower bounds for the fully dynamic case where $G_t$ can be arbitrary. 

We obtain the first competitive algorithms for \textit{capacitated} online recoloring and fully dynamic recoloring, in which there is a bound on the number or weight of vertices in each color.  Our first set of results is for $2$-recoloring using algorithms that are $(1+\varepsilon)$-resource augmented where $\varepsilon\in (0,1)$ is an arbitrarily small constant.  Our main result is an $O(\log n)$-competitive \textit{deterministic} algorithm for weighted bipartite graphs, which is \textit{asymptotically optimal} in light of an $\Omega(\log n)$ lower bound that holds for an unbounded amount of augmentation.  We also present an $O(n\log n)$-competitive deterministic algorithm for \textit{fully dynamic} recoloring, which is optimal within an $O(\log n)$ factor in light of a $\Omega(n)$ lower bound that holds for an unbounded amount of augmentation.

Our second set of results is for $\Delta$-recoloring in an $(1+\varepsilon)$-overprovisioned setting where the maximum degree of $G_t$ is bounded by $(1-\varepsilon)\Delta$ for all $t$, and each color assigned to at most $(1+\varepsilon)\frac{n}{\Delta}$ vertices, for an arbitrary $\varepsilon > 0$.  Our main result is an $O(1)$-competitive randomized algorithm for $\Delta = O(\sqrt{n/\log n})$. We also present an $O(\Delta)$-competitive deterministic algorithm for $\Delta \le \varepsilon n/2$.  Both results are \emph{asymptotically optimal}.
\end{abstract}

\section{Introduction}
The challenge of efficiently allocating virtual machines (VMs) across large clusters in a cloud or data center is well documented \cite{innovationsfordatacenter, silva2018approachesforoptvmplacementandmigration, vmmigrationplanninginsdns}. Most big-data applications are inherently distributed. To minimize congestion, communication costs, and service delays incurred due to inter-VM and inter-process communication, distributed systems are demand-aware \cite{avin2020demand}, requiring scheduling algorithms to efficiently adapt to network traffic. While co-locating frequently communicating VMs and processes in a single cluster (corresponding to affinity requests in cloud or VM environment) is desirable to minimize communication overhead, service providers are also required to satisfy anti-affinity requests \cite{anti-aff,antiaff-2}, which require that certain VMs must not be co-located. Large-scale deployment systems including Kubernetes and VMWare support anti-affinity rules to provide safety, performance and robustness \cite{kubernetes, vmware}.

The online recoloring problem (also referred to as online disengagement) was introduced recently by Azar, Machluf, Patt-Shamir and Touitou \cite{azar_et_al:LIPIcs.ICALP.2022.13} to investigate algorithmic challenges that arise in scheduling of VMs in distributed systems and cloud computing. In this problem, we have a graph $G$ on a fixed set of $n$ vertices (corresponding to VMs) with weights (corresponding to VM sizes) and edges (corresponding to anti-affinity requests) are revealed over time. The goal is to ensure that endpoints of every edge encountered are assigned a different color (i.e. the corresponding VMs scheduled on different clusters). The graph induced by the requested edges is assumed to be $k$-colorable (colors corresponding to $k$ clusters) and the objective is to minimize the total weight of vertices recolored throughout time. The cost of recoloring any vertex captures the size and migration cost of the corresponding VM. 

In \cite{azar_et_al:LIPIcs.ICALP.2022.13}, the authors study the \textit{uncapacitated} version where the total weight of vertices assigned a color is unconstrained. Since vertex coloring is NP-hard \cite{karp2010reducibility}, they focus on polynomially tractable cases when the graph induced by the request sequence is either bipartite ($k=2$) or has maximum degree at most $\Delta$ ($k=\Delta+1$). In the static setting, both cases admit linear-time \textit{greedy algorithms}. In the online setting, while existence of a proper $k$-coloring guarantees that any vertex needs to be recolored at most once, the uncertainty of future requests can force multiple recolorings. The \textit{capacitated} online recoloring problem captures the practical scenario where cluster sizes are constrained, i.e. the total weight of vertices assigned any color is bounded. Resolving this was left as an open problem in \cite{azar_et_al:LIPIcs.ICALP.2022.13}. 
\subparagraph{Online Recoloring.} We formally describe the online recoloring problem \cite{azar_et_al:LIPIcs.ICALP.2022.13}. Given a fixed set of $n$ vertices denoted by $V$, weight function $w:\,V\rightarrow \mathbb{R}^+$ and an initial coloring $c_0:\, V\rightarrow [k]$ (where $[k]=\{1,2,...,k\}$), an algorithm maintains a $k$-coloring $c:\,\,V\rightarrow [k]$ for $V$ under an online sequence of requests, $\sigma=\{\sigma_i\}_{i=1}^T$ where $\sigma_i$ is an edge $(u_i, v_i)$ between $u_i,v_i\in V$. Let $\sigma_{\leq t}=\{\sigma_i\}_{i=1}^t$. For all $t$, $c$ must be a proper $k$-coloring for $G_t$, the subgraph induced by $\sigma_{\leq t}$, i.e. for all $(u_i, v_i)$ where $i\leq t$, $c(u_i)\neq c(v_i)$. At any time, any vertex $v$ can be recolored for a cost $w(v)$. In the unweighted setting, $w(v)=1$ for all $v\in V$.
The objective is to minimize the total recoloring cost incurred for the request sequence $\sigma$. 

We analyze online algorithms in the standard competitive analysis framework. 
Let the total cost by an optimal algorithm (resp., an online algorithm $\mathcal{A}$) on a request sequence $\sigma$ be denoted by $OPT_{\sigma}$ (resp. $\mathcal{A}_{\sigma}$).  
An algorithm $\mathcal{A}$ for online recoloring is $\alpha$-competitive if for any request sequence $\sigma$, $\mathcal{A}_{\sigma}\leq \alpha OPT_{\sigma}$.  Note that $\sigma$ can be assumed to be finite w.l.o.g. (in particular, $|\sigma|=O(n^2)$) since a repeated edge request incurs no cost.

\subparagraph{Fully Dynamic Recoloring.} In the above formulation, the graph induced by the request sequence is assumed to be $k$-colorable. We also consider a \textit{fully dynamic} variant \cite{azar2023competitive}, which captures the case when requests are temporary, i.e, the graph induced by the request sequence is not necessarily $k$-colorable. For all $t$, it is \textit{only} required to maintain that $c(u_t)\neq c(v_t)$ for the current request $\sigma_t=(u_t, v_t)$. Note that the request sequence can be unbounded in this model. An algorithm $\mathcal{A}$ for fully dynamic recoloring is $\alpha$-competitive if for any request sequence $\sigma$, $\mathcal{A}_{\sigma}\leq \alpha OPT_{\sigma}+ \mu$ where $\mu$ is a constant independent of request sequence $\sigma$.

\enlargethispage{\baselineskip}
\subparagraph{Capacitated Online and Fully Dynamic Recoloring.} For \textit{uncapacitated} recoloring described above, the total weight of vertices assigned a particular color is unconstrained. In \emph{capacitated} recoloring, the total weight of vertices assigned color $i$ is at most $W_i$ for all $i$ at any time $t$. In this paper, we focus our attention on the \textit{uniform} capacity case, i.e. $W_i=B$, for all $i\in [k]$. In the unweighted case, $B$ is the \textit{number} of vertices assigned any color $i\in [k]$. For capacitated online and fully dynamic recoloring, we assume that the graph induced by $\sigma$ admits a proper $k$-coloring while respecting capacity constraints for each $i\in [k]$.

\subparagraph{Resource Augmentation.} Our algorithms for online and fully dynamic $2$-recoloring \textit{mildly utilize} resource augmentation: we assume the total weight of vertices that can be assigned color $i$ is $B(1+\varepsilon)$ $\forall i\in [k]$ for an arbitrarily small constant $0<\varepsilon<1$. The competitivene ratio is analyzed with respect to an offline algorithm constrained to a capacity of $B$ for all $i$. Nevertheless, our deterministic and randomized lower bounds hold for resource augmented algorithms--thus, establishing tightness of the upper bounds we obtain. 

\subparagraph{Overprovisioned Setting.} For $\Delta$-recoloring, we introduce and study an \emph{overprovisioned} setting where the graph induced by $\sigma$ has maximum degree $(1 - \varepsilon)\Delta$ and each color $i\in [\Delta]$ has capacity $(1 + \varepsilon)n/\Delta$. In contrast to the resource augmented framework, the competitive ratio is analyzed with respect to a offline algorithm with the same number of colors and unconstrained capacity. Remarkably, even in this pessimistic analytical framework, we obtain tight deterministic and randomized algorithms.

\subsection{Related Work}
\subparagraph{Online Recoloring.} Azar el al. \cite{azar_et_al:LIPIcs.ICALP.2022.13} recently introduced the online (uncapacitated) recoloring problem. They give optimal deterministic and randomized algorithms for online recoloring when $k=2$ and $k=\Delta+1$. For the bipartite case, they give an optimal $O(\log n)$-competitive deterministic algorithm and a lower bound of $\Omega(\log n)$, which holds for randomized algorithms. For $(\Delta+1)$-recoloring, they give a deterministic $O(\Delta)$-competitive algorithm, a $O(\log \Delta)$-competitive randomized algorithm, and matching lower bounds. They introduce fully dynamic recoloring and consider three equivalent edge insertion models and show that fully dynamic recoloring has substantially worse competitive ratios than online recoloring: for $2$-coloring, they consider edge requests on a length $n$-odd cycle (which cannot be properly colored in the online recoloring setting) and provide a reduction from the metrical task systems (MTS) problem on a length $n$-odd cycle, yielding $\Omega(n)$ deterministic and $\Omega (\frac{\log n}{\log \log n})$ randomized lower bounds respectively for fully dynamic recoloring.  

In subsequent work \cite{azar2023competitive}, Azar et al.\ give general lower bounds for $k$-coloring. They provide $\Omega(k\log (n/k)$ deterministic and $\Omega (\log k\cdot \log (n/k))$ and randomized lower bounds. They also present lower bounds for capacitated $k$-coloring, including an $\Omega(n)$ deterministic and an $\Omega(\log n)$ randomized lower bound, for the \textit{tight} case when $kB=\sum_{v\in V}w(v)$. Note that all lower bounds given in \cite{azar2023competitive, azar_et_al:LIPIcs.ICALP.2022.13} hold even for \textit{unweighted} instances, i.e. when $w(v)=1,\, \forall v\in V$.

\subparagraph{Dynamic Coloring.} Maintaining a proper coloring in fully dynamic graphs is well studied \cite{henzinger2020explicit, henzingerpeng}. Since determining the chromatic number number of a graph is NP-hard \cite{karp2010reducibility}, most works focus on restricted classes of graphs, bounded chromatic number or maximum degree. In dynamic coloring, an edge (or vertex) is either inserted or deleted from the graph and the algorithm is required to quickly re-compute a proper coloring. The objective is to minimize the update time \cite{bhattacharya2022fully, wein}, and/or the recourse \cite{10.1007/s00453-022-01050-7, barba2017dynamic} (i.e. the number of recolored vertices after an update) in a worst-case or amortized sense. Kashyop et al.~\cite{10.1007/s00453-022-01050-7},  Barba et al.~\cite{barba2017dynamic} and Bosek et al.~\cite{bosek_et_al:LIPIcs.SWAT.2020.17} study trade-offs between the number of colors, number of recolorings and update time for dynamic coloring for various classes of graphs (including bipartite, bounded degree and arboricity, and interval graphs). 
The crucial difference between dynamic and online recoloring is the fact that the number of recolorings in the former do not guarantee competitiveness in the latter; a dynamic algorithm incurring $O(1)$ recolorings per update only guarantees $O(T)$ recolorings where $T=|\sigma|$ in the online setting. When $OPT_{\sigma}\ll T$, this precludes competitiveness which is measured against $OPT_{\sigma}$. 
\subparagraph{Online Balanced Graph Partitioning.} Online balanced graph partitioning (OBGR) introduced by Avin et al., \cite{avin2020dynamic} and later explored in \cite{racke2022approximate,henzinger2021tight, pacut2021optimal, rajaraman2022improved} is related to online recoloring. While the objective of recoloring is to optimize allocations of VM's under \textit{anti-affinity} requests, OBGR deals with \textit{affinity} requests. In OBGR, given a graph $G$ on $n$ vertices that are initially assigned to $\ell$ servers each of capacity $k$, for any online request $(u,v)$, an algorithm incurs a unit cost if $u$ and $v$ are in different clusters and 0 otherwise. Any vertex can be migrated for a unit cost at any time. The goal is to minimize the total migration and communication costs incurred for the request sequence. Since the static version of balanced graph partitioning is known to be NP-hard even for $\ell= 2$ (the minimum bisection problem) \cite{garey1974some} resource augmentation and randomization have been employed to obtain competitive algorithms (some taking exponential time) for OBGR. Avin et al.~\cite{avin2020dynamic} give an $O(k^2\ell^2)$ competitive algorithm without resource augmentation, and a $O(k\log k)$ exponential time algorithm with $(2+\varepsilon)$-augmentation (improved in \cite{forner2021online} with $(1+\varepsilon)$-augmentation). For the \textit{learning} variant of OBGR, Henzinger et al.~\cite{henzinger2021tight} give optimal deterministic $O(\ell\log k)$- and randomized $O(\log k+\log\ell)$-competitive algorithms. R\"acke et al~\cite{racke2022approximate} give an $O(\log n)$-approximation algorithm with $(2+\varepsilon)$-augmentation for the \textit{offline} case. The current best deterministic online algorithm for OBGR \cite{rajaraman2022improved} obtains $O(k\ell\log k)$-competitiveness with $(1+\varepsilon)$-augmentation. 

\subsection{Our results}
We present the first competitive algorithms for \textit{capacitated} online 2-recoloring, making progress on open problems posed by Azar et al.~\cite{azar_et_al:LIPIcs.ICALP.2022.13}.  We also introduce the capacitated $\Delta$-recoloring problem in an overprovisioned setting and obtain asymptotically optimal algorithms.

\subparagraph{Capacitated \texorpdfstring{\boldmath $2$}{2}-Recoloring.}
Our algorithms for capacitated 2-recoloring hold with $(1+\varepsilon)$-resource augmentation, i.e. the algorithms utilize a capacity of $W=(1+\varepsilon)B$ for each cluster, while the optimal offline algorithm has a capacity of $B$ where $B=\frac{\sum_{v\in V}w(v)}{2}$, (w.l.o.g., assume $\sum_{v\in V}w(v)$ is even) and $\varepsilon\in (0,1)$ is a small constant.\\
\underline{\textit{Fully Dynamic Recoloring}:} Our first result is an algorithm for \textit{fully dynamic} capacitated recoloring which works for unweighted instances. 
\begin{restatable}{theorem}{fdthm}
\label{fdthm}
There exists a deterministic $O(n\log n)$-competitive $(1+\varepsilon)$-resource augmented algorithm for fully dynamic capacitated online 2-recoloring that works for unweighted instances.
\end{restatable}
We give a deterministic lower bound for fully dynamic $2$-recoloring that holds for an unweighted instance and \textit{unbounded} resource augmentation, establishing that our algorithm is tight within a $O(\log n)$ factor. The proof of Theorem \ref{lowerbdet} is deferred to Appendix \ref{app: lbdet}.
\begin{restatable}{theorem}{lowerbdet}
\label{lowerbdet}
The competitive ratio of any deterministic algorithm for capacitated fully dynamic $2$-recoloring is $\Omega(n)$, for any resource augmentation $\varepsilon>0$ and unweighted instances.
\end{restatable}

\noindent\underline{\textit{Online Recoloring}:}
Our main result is an asymptotically \textit{optimal} deterministic $O(\log n)$-competitive algorithm for capacitated online $2$-recoloring which works for \emph{weighted} instances.

\begin{restatable}{theorem}{logncompetitive}\label{thm: logncompetitive}
There exists an optimal deterministic $O(\log n)$-competitive $(1+\varepsilon)$-resource augmented algorithm for capacitated online $2$-recoloring that works for weighted instances.
\end{restatable}
We complement our result giving a lower bound for the capacitated case similar to the uncapacitated case in \cite{azar_et_al:LIPIcs.ICALP.2022.13}, which holds for an \textit{unbounded} amount of resource augmentation, $\varepsilon>0$. The proof of Theorem \ref{lowerbrand} is deferred to Appendix \ref{app: randlognlb}.
\begin{restatable}{theorem}{lowerbrand}
\label{lowerbrand}
The competitive ratio of any randomized algorithm for capacitated online $2$-recoloring is $\Omega(\log n)$, which holds for an arbitrary amount of resource augmentation $\varepsilon>0$.
\end{restatable}

\subparagraph{Capacitated \texorpdfstring{\boldmath $\Delta$}{Delta}-recoloring.}
We introduce the \emph{$(1 + \varepsilon)$-overprovisioned} setting for capacitated $\Delta$-recoloring, and give optimal deterministic and randomized algorithms. In this setting, there are $n$ vertices, and $\Delta$ colors. Each color is assigned exactly $\frac{n}{\Delta}$ vertices initially. The maximum degree of the graph induced by the request sequence is bounded by $(1-\varepsilon)\Delta$, and the total capacity per color is $(1+\varepsilon)\frac{n}{\Delta}$. The request sequence is finite. In contrast to the resource augmented setting for which the competitiveness is analyzed with respect to an unaugmented optimal offline algorithm, we analyze competitiveness with respect to an optimal offline algorithm with the same number of colors and unrestricted capacity per color.

We first note that a lower bound of $\Delta$ holds for any deterministic algorithm in this setting. The instance is similar to the one given for uncapacitated $(\Delta+1)$-recoloring given in \cite{azar_et_al:LIPIcs.ICALP.2022.13}. Thus, deterministic $(1+\varepsilon)$-overprovisioned $\Delta$-recoloring is as hard as deterministic $(\Delta+1)$ recoloring. The proof of Theorem \ref{lowerbdelta} is deferred Appendix \ref{app: lbdelta}.

\begin{restatable}{theorem}{lowerbdelta}
\label{lowerbdelta}
The competitive ratio of any deterministic algorithm for $(1+\varepsilon$)-overprovisioned $\Delta$-recoloring is $\Omega(\Delta)$.
\end{restatable}
We give an optimal deterministic algorithm matching the above lower bound. 

\begin{restatable}{theorem}{deltaOne}
\label{thm:delta-1}\sloppy
There exists a $O(\Delta)$-competitive deterministic algorithm for $(1+\varepsilon$)-overprovisioned $\Delta$-recoloring when $\Delta \le \eps n/2$.
\end{restatable}

Our main result for $\Delta$-recoloring is an optimal $O(1)$-competitive randomized algorithm. 

\begin{restatable}{theorem}{deltaTwo}
\label{thm:delta-2}
There exists a $O(1)$-competitive randomized algorithm for $(1+\varepsilon$)-overpro\-visioned $\Delta$-recoloring that works against an oblivious adversary when $\Delta=O\left(\sqrt{\frac{n}{\log n}}\right)$.
\end{restatable}

\subsection{Technical Overview}
\subparagraph{Algorithms for 2-Recoloring.} Our algorithms maintain a set of bipartite components induced by the request sequence. We develop a subroutine $\texttt{Rebalance}$ which is a FPTAS, to periodically assign components while respecting capacity constraints.

Our fully dynamic algorithm is a phase-based algorithm where the graph induced by requests in any phase is bipartite and admits a coloring such that the total weight of vertices assigned any color is $W=(1+\varepsilon)B$. During a phase, our algorithm works as follows: on a request $(u,v)$ where $u$ is in component $P_1$ and $v$ in $P_2$, s.t. $P_1$ is the heavier of the two components, $P_2$ is recolored and merged to $P_1$. If recoloring of $P_2$ leads to capacity violation, \texttt{Rebalance} is invoked. The total recoloring cost charged to a vertex per phase is bounded by separately considering when it is part of a heavy or light component. We observe that whenever recoloring of a light component leads to a call to $\texttt{Rebalance}$, vertices in light components of total weight at least $\Omega(W)$ must have been recolored since the last call to $\texttt{Rebalance}$ (if any). On the other hand, when two heavy components are merged, $\texttt{Rebalance}$ is always invoked; while a single call to \texttt{Rebalance} incurs $O(W)$ cost, the number of such calls are only $O(1)$ per phase. Combined with a charging argument which bounds the number of times any vertex can be recolored in a phase by $O(\log n)$, we obtain a $O(n\log n)$-competitive algorithm for fully dynamic recoloring. We complement this result with a lower bound of $\Omega(n)$ which holds for an unbounded amount of augmentation, establishing that our algorithm is tight to within a $O(\log n)$ factor.

Our deterministic $O(\log n)$-competitive algorithm uses the fully dynamic algorithm as a subroutine once a lower bound of $\Omega(W)$ on $OPT_{\sigma}$ has been determined. Our algorithm incorporates \textit{laziness} together with the \textit{greedy} approach. While it is crucial to maintain a coloring that is \textit{close} to the initial coloring for any component to be competitive, we recolor a component according to an optimal coloring only after a constant factor increase in its weight, until it is possibly merged to a heavier component. Whenever a light component $P_2$ is merged to a heavy component $P_1$ such that the weight of $P_1$ does not change significantly, only $P_2$ is recolored (if necessary). A crucial property we exploit in the analysis is the monotonicity of the distance to initial coloring with respect to the weight of any component. As a result, the cost of periodically re-coloring components to mimic an optimal coloring after weight increases can be charged to $OPT_{\sigma}$. An intricate charging argument allows us to bound the cost of both periodic recolorings and greedy recolorings by $OPT_{\sigma}O(\log n)$.

If the request sequence precludes optimal colorings without significant rebalancing (i.e. when a lower bound of $OPT_{\sigma}=\Omega(W)$ is determined), our algorithm transitions to the fully dynamic algorithm. The total cost of the fully dynamic algorithm for the remaining request sequence is bounded by $O(W\log n)$. Note that our fully dynamic algorithm is $O(n\log n)$-competitive. Our analysis reveals that in the case when $OPT_{\sigma}=\Omega(W)$, the fully dynamic algorithm yields $O(\log n)$ competitiveness on \textit{weighted} instances. We give a \emph{randomized} lower bound of $O(\log n)$ for online recoloring which holds for an algorithm with \emph{unbounded} resource augmentation, establishing that our algorithm for online $2$-recoloring is \textit{asymptotically optimal}.

\subparagraph{Algorithms for \texorpdfstring{\boldmath $\Delta$}{Delta}-Recoloring.}
We give \textit{asymptotically optimal} deterministic and randomized algorithms for $\Delta$-recoloring in the $(1+\varepsilon)$-overprovisioned setting. Both algorithms consist of a \textit{recoloring} subroutine that colors an individual vertex and a \textit{rebalancing} subroutine when any color reaches capacity. The implementations of these subroutines are different in the two algorithms. Similar to the approach in \cite{azar_et_al:LIPIcs.ICALP.2022.13}, we lower bound $OPT_{\sigma}$ as the size of the minimum vertex cover on the graph induced by the set of monochromatic edges in $\sigma$ with respect to the initial coloring. Both algorithms maintain a $2$-approximate vertex cover $C$. If the requested edge is monochromatic, both algorithms recolor only the endpoint in $C$. 

Our deterministic algorithm picks a feasible color assigned to the least number of vertices, ensuring that whenever all feasible colors are \textit{full} for a vertex, the rebalancing cost of $O(n)$ can be charged to the increase in vertex degrees. To rebalance, we utilize an algorithm for equitable coloring of graphs with bounded maximum degree by Kierstead et al. \cite{kierstead2010fast}. We complement our results by giving a $O(\Delta)$ lower bound, similar to an instance in \cite{azar_et_al:LIPIcs.ICALP.2022.13}.

Our $O(1)$-competitive randomized algorithm uses a simple rebalancing procedure which sequentially colors vertices randomly among their feasible colors, where a color $c$ is feasible for a vertex $v$ if none of $v$'s neighbors are assigned $c$. To implement the recoloring subroutine for $v$, a random feasible color for $v$ is picked. While our randomized algorithm is simple, we face multiple technical challenges towards obtaining competitiveness since straightforward concentration inequalities cannot be applied, because recolorings (and rebalancings) involve \textit{dependent} events. For instance, our analysis of rebalancing utilizes a martingale to derive a concentration bound on the load of any color. We also invoke stochastic domination arguments to upper bound the load on any color due to recoloring, and lower bound the number of recolorings performed before any color is overloaded. \subsection{Open Problems} While we obtain tight bounds for several variants of capacitated recoloring, our work leaves several open problems. Obtaining competitive algorithms for fully dynamic 2-recoloring, and $\Delta$-recoloring in the weighted setting remains open. Bridging the $\Theta(\log n)$ gap between the upper and lower bounds for fully dynamic 2-recoloring is a natural question. The case of fully dynamic $\Delta$-recoloring is wide open, even in the uncapacitated setting. Another direction is the case of non-uniform capacities; we believe that Theorem \ref{thm: logncompetitive} extends to this case. 

The overprovisioned model for $\Delta$-recoloring allows for stronger results than the resource augmented model in that the offline algorithm is permitted to utilize the same resources. Our current results assume over-provisioning in colors (with respect to the maximum degree) and capacity per color; it would be insightful to obtain results where either the capacity, or the number of colors is overprovisioned. Also, whether the $O(\log n)$-competitive 2-recoloring result generalizes to the overprovisioned model would be interesting to resolve. Finally, can the range of $\Delta$ for our results in the over-provisioned model be improved?

\section{Capacitated \texorpdfstring{\boldmath $2$}{2}-Recoloring}

\subsection{Preliminaries}\label{sec: prelim}
Let $[n]$ denote the set of integers $\{1,2,...,n\}$. We denote the two clusters as $C_1$ and $C_2$, each of which has capacity $W=\frac{1+\varepsilon}{2}\sum_{v\in V} w(v)$. A vertex $v$ (resp. set $S\subseteq V$) is said to be \textit{scheduled} on a cluster $C\in \{C_1,C_2\}$ if $v\in C$ (resp. $S\subseteq C$). At any point, let $N(C)$ denote the \textit{residual capacity} of cluster $C\in \{C_1, C_2\}$ where $N(C)=W-\sum\limits_{v\in C}w(v)$. We assume that $w(v)\in [\frac{1}{2}\sum_{v\in V} w(v)]$ for all $v\in V$.

Our algorithms maintain a set of \emph{bipartite connected components} $\mathcal{P}$ induced by the request sequence. A connected component $P$ is a maximal set of vertices such that for any $u_i, v_i\in P$, there is a path between $u_i$ and $v_i$ consisting only of vertices in $P$. Initially, $\mathcal{P}=\{\{v\}|\, v\in V\}$ so that each component consists of a single vertex. Each component $P_i\in \mathcal{P}$ for $i\in [|\mathcal{P}|]$ with bi-partition $(A_i, B_i)$ is expressed as $P_i=(A_i, B_i)$. We refer to the weight (resp. size) of the component as $w(P_i)$ (resp. $|P_i|$), where $w(P_i)=\sum\limits_{v\in A_i\cup B_i} w(v)$ (resp. $|P_i|=|A_i|+|B_i|$). 
\subparagraph{Component Merges.}A request $(u_t, v_t)$ between vertices in distinct components $P_1=(A_1, B_1)$ and $P_2=(A_2, B_2)$ leads to a merge of $P_1$ and $P_2$ into a new component $P_3$, following which $\mathcal{P}$ is updated. More precisely, if $u_t\in A_1$ (resp.\ $B_1$) and $v_t\in A_2$, then $P_3=(A_1\cup B_2, A_2\cup B_1)$ (resp.\ $P_3=(A_1\cup A_2, B_1\cup B_2)$), and analogously for the case when $v_t\in A_1$. For any component $P_i\in \mathcal{P}$, our algorithms schedule $A_i$ and $B_i$ in different clusters. Observe that $|\mathcal{P}|\leq n$ and $\max_{P_i\in \mathcal{P}}(|P_i|)\leq n-1$.

\subparagraph{Recolor Subroutine.} Our algorithms employ a recoloring subroutine $\texttt{Recolor}$. Given a current coloring $c(v)=j$ for any vertex $v\in V$, \texttt{Recolor}($v, i$) recolors $v$ to $i$ if the current coloring $j\neq i$ and sets $N(C_j)$ (resp. $N(C_i)$) to $N(C_j)-w(v)$ (resp. $N(C_i)+w(v)$). 

\subparagraph{Rebalancing Subroutine.}Our algorithms periodically call a re-balancing subroutine\linebreak $\texttt{Rebalance}$ which takes as input a weight parameter $W'$, component set $\mathcal{P}$, and a parameter $\varepsilon$. It computes an \textit{assignment} of bipartite components in $\mathcal{P}$ to clusters such that the total weight of vertices assigned on $C_1$ is at most $W'(1+\varepsilon)$. If this is not possible, it terminates. While an assignment satisfying weight constraint exactly $W'$ can be determined in pseudo-polynomial time $O(|\mathcal{P}|W')=O(nW')$, \texttt{Rebalance} is a FPTAS (fully polynomial-time approximation scheme) and takes $O(\frac{n^2 \ln W'}{\varepsilon})$ time. The full algorithm is deferred to Appendix \ref{sec: Apprebalance}. 

\subsection{An \texorpdfstring{\boldmath $O(n\log n)$}{O(nlog n)}-Competitive Algorithm for Fully Dynamic Recoloring}
\label{sec: fullydynamic}
We present an algorithm \texttt{Greedy-Recoloring} for fully dynamic recoloring in the unweighted setting with $(1+\varepsilon)$-augmentation, where $\frac{8}{n}\leq \varepsilon<1$. We begin by defining a phase.

\begin{definition}
A phase $\mathcal{R}$ of request sequence $\sigma$ is a maximal contiguous sub-sequence of $\sigma$ such that, i) all components induced by requests in the $\mathcal{R}$ are bipartite and, ii) there exists a feasible assignment of components such that the total number of vertices assigned to cluster $C_1$ is at most $(1+\frac{\varepsilon}{2})\frac{n}{2}$.
\end{definition}

The sequence $\sigma$ is partitioned into consecutive phases $\mathcal{R}_1, \mathcal{R}_2,...$ and our algorithm initializes each phase $\mathcal{R}_i$ with a set of singleton components $\mathcal{P}_i=\{\{v\}|\, v\in V\}$ and assigns exactly $\frac{n}{2}$ vertices to each cluster. The set of components $\mathcal{P}_i$ induced by requests throughout any phase $\mathcal{R}_i$ is maintained. A phase terminates when either i) a request $(u_t,v_t)$ is encountered where $u_t, v_t\in A_j$ or $u_t,v_t\in B_j$ for any component $P_j=(A_j, B_j)$ in $\mathcal{P}_i$, i.e. $P_j$ ceases to be bipartite or ii) invoking $\texttt{Rebalance}(\mathcal{P}_i, \frac{n}{2}, \frac{\varepsilon}{2})$ returns infeasible, i.e. an assignment of components is not possible while fulfilling the capacity constraint of $(1+\frac{\varepsilon}{2})\frac{n}{2}$. Thereafter, the next phase $\mathcal{R}_{i+1}$ begins. The following lemma follows from the definition of a phase.

\begin{restatable}{lemma}{optlb}
\label{opt-lb}
For an unweighted instance such that $w(v)=1$ for all $v\in V$, OPT incurs a cost of at least 1 per phase. 
\end{restatable}

\begin{proof}
If $\texttt{Rebalance}$ does not yield an assignment of components induced by requests in phase $\mathcal{R}_i$, OPT must incur a cost of at least 1 during phase $\mathcal{R}_i$. If $\mathcal{R}_i$ terminates because a component ceases to be bipartite, there exist requests in $\mathcal{R}_i$ on an odd cycle for which a single recoloring is necessary. 
\end{proof}
\subparagraph{Algorithm \texttt{Greedy-Recoloring}.} We give a high level description of our Algorithm \texttt{Greedy-}\linebreak\texttt{Recoloring}. The pseudo code is given in Appendix \ref{app: fullydynamic}.
Consider a request $(u_t,v_t)$ in phase $\mathcal{R}_i$. If $u_t, v_t\in P$ for some component $P\in \mathcal{P}_i$ such that $P$ stays bipartite, nothing is done. Else if $P$ ceases to be bipartite, phase $\mathcal{R}_i$ terminates.

If $u_t, v_t$ are in distinct components $P_1=(A_1, B_1), P_2=(A_2, B_2)$ (w.l.o.g., assume that $|P_1|\geq |P_2|$), vertices in $A_1\cup A_2$ are assigned to $C_1$ and vertices in $B_1\cup B_2$ are assigned to $C_2$, $P_2$ is merged into $P_1$. If $|P_2|\leq \frac{\varepsilon n}{8}$ and recoloring is possible while fulfilling capacity constraints, $P_2$ is re-colored. More concretely, if $u_t\in A_1$ (resp. $B_1$) and $v_t\in A_2$ (resp, $B_2$), vertices in $B_2$ (resp. $A_2$) are recolored and assigned to $C_1$ (resp. $C_2$), and vertices in $A_2$ (resp. $B_1$) are recolored and assigned to $C_2$ (resp. $C_1$) if residual capacities $N(C_1), N(C_2)$ are sufficient. Otherwise, \texttt{Rebalance}($\mathcal{P}_i, \frac{n}{2},\frac{\varepsilon}{2}$) is invoked. If \texttt{Rebalance} yields a feasible component assignment, vertices are appropriately recolored such that $N(C_1)$ and $N(C_2)$ are both at least $\frac{\varepsilon}{2}\cdot\frac{n}{2}=\frac{\varepsilon n}{4}$. 

\begin{proof}[Proof of Theorem~\ref{fdthm}]
By virtue of the recolorings performed, Algorithm\!$\texttt{Greedy-Recoloring}$ always maintains a 2-coloring of the graph induced by the sequence of requests in a phase. Note that \texttt{Rebalance} incurs a total recoloring cost at most $n$. Before the first request arrives, we call $\texttt{Greedy-Recoloring}(\{v|\,v\in V\}, \frac{n}{2})$  For the sake of analysis, we say a component $P_j\in \mathcal{P}_i$ during phase $\mathcal{R}_i$ is \textit{small} if $|P_j|\leq \frac{\varepsilon n}{8}$ and \textit{large} otherwise. Note that the algorithm always recolors the smaller component, $P_2$ when $P_1=(A_1, B_1)$ and $P_2=(A_2, B_2)$ are merged, as long as component $P_2$ is small and capacity constraints are \textit{not violated}. In this case, $|P_1|\geq 2|P_2|$, and each vertex in $P_2$ is charged a recoloring cost of 1. If $P_2$ is large, $\texttt{Rebalance}$ is called. 

The key observation is that whenever a small component $P_1$ is created as a result of two small components $P_1$ and $P_2$ being merged, and $A_2$ and $B_2$ cannot be recolored because either $N(C_1)-|A_2|<|B_2|$ or $N(C_2)-|B_2|<|A_2|$, at least $\frac{\varepsilon n}{8}$ vertices (let $S$ denote the set of these vertices) must have been successfully recolored since the last time \texttt{Rebalance} was called. This holds since $(1+\varepsilon)\frac{n}{2}-(1+\frac{\varepsilon}{2})\frac{n}{2}-|P_1|\geq \frac{\varepsilon n}{8}$. We charge the recoloring cost of $n$ incurred as a result of \texttt{Rebalance} uniformly to all vertices in $S$, i.e. every vertex in $S$ receives a charge of $\frac{n}{\varepsilon n/8}=\frac{8}{\varepsilon}$. Each time a vertex in this manner, the component containing it doubles in size. Hence, any vertex is charged at most $\frac{8}{\varepsilon}\log \frac{\varepsilon n}{8}=O(\log n)$ while it is part of a small component. 

On the other hand, if $P_2$ is large $\texttt{Rebalance}$ is called. The recoloring cost is charged uniformly to all vertices in $P_2$, so that each vertex in $P_2$ receives a charge of at least $\frac{n}{\varepsilon n/8}=\frac{8}{\varepsilon}$. Thus, every time a vertex is part of a large component it can be charged a recoloring cost at most $\log (\frac{8}{\varepsilon})=O(1)$ since its component size doubles and the maximum size of any component is $n$. 

If recoloring does not lead to a call to \texttt{Rebalance}, vertices in $P_2$ are simply charged 1 for the recoloring. Thus, every vertex receives a total charge of $\log n+ \frac{8}{\varepsilon}\log (\frac{8}{\varepsilon})+\frac{8}{\varepsilon}\log \frac{\varepsilon n}{4}=O(\log n)$ throughout the phase. The total recoloring cost incurred by the algorithm is $O(n\log n)$, and by Lemma \ref{opt-lb}, this yields $O(n\log n)$-competitiveness.
\end{proof}

\subsection{An \texorpdfstring{\boldmath $O(\log n)$}{O(log n)}-Competitive Algorithm for Online Recoloring}
\label{sec: logncompetitive}
In this section, we present our main result for 2-recoloring which is an asymptotically optimal $O(\log n)$-competitive algorithm. Recall that in the uniform capacity setting, the graph induced by request sequence $\sigma$ admits a proper coloring such that the weight of vertices assigned any color is exactly $\frac{1}{2}\sum_{v\in V}w(v)$.   
\subparagraph{Notation and Definitions.} Given an initial coloring $c_0$, an arbitrary coloring $c$ and $S\subseteq V$, let $d_S(c, c_0)$ denote the weighted Hamming distance of $c$ from $c_0$ restricted to $S$, i.e. $d_S(c, c_0)=\sum\limits_{v\in S:\, c_0(v)\neq c(v)} w(v)$. Let $d_S(c)\coloneqq d_S(c,c_0)$ for brevity. For a bipartite component $P_i=(A_i, B_i)$, there are exactly two feasible colorings; let $c_1$ denote the coloring for which vertices in $A_i$ (resp. $B_i$) are colored $1$ (resp $2$) and $c_2$ denote the coloring for which vertices in $A_i$ (resp. $B_i$) are colored $2$ (resp. $1$). Let $c_m$ denote the optimal coloring where $m\in \{1,2\}$, such that $d_{P_i}(c_m)=\min\{d_{P_i}(c_1), d_{P_i}(c_2)\}$, i.e. $c_m$ corresponds to a coloring incurring the minimum recoloring cost for $P_i$ w.r.t. the initial coloring $c_0$. Since there are only two possible colorings, it follows that $d_{P_i}(c_m)\leq \frac{w(P_i)}{2}$. Our algorithm assigns vertices in $A_i$ on cluster $C_m$ and vertices in $B_i$ on cluster $C_{3-m}$.
Finally, let $E(P_i)$ denote the \textit{estimated weight} of component $P_i$. Our algorithm maintains the property that $E(P_i)\geq \frac{w(P_i)}{(1+\frac{\varepsilon}{4})}$, and periodically recomputes optimal colorings whenever the total weight $w(P_i)>(1+\frac{\varepsilon}{4})E(P_i)$, and updates $E(P_i)$ to $w(P_i)$. Let $c(S)$ denote the restriction of some coloring $c$ to set $S\subseteq V$.
\subparagraph{Algorithm \texttt{Follow-Greedy}.} The algorithm begins with the set of components $\mathcal{P}=\{\{v\} v\in V\}$ and $E(\{v\})=w(v)$ for all $v$. For any request $(u_t, v_t)$ such that $u_t\in P_1$, $v_t\in P_2$ where w.l.o.g., $w(P_1)\geq w(P_2)$, the algorithm merges $P_2$ into $P_1$. There are two cases to consider: i) if $c(u_t)=c(v_t)$ and $w(P_1)\leq E(P_1)(1+\frac{\varepsilon}{4})$ following the merge and, ii) $w(P_1)>E(P_1)(1+\frac{\varepsilon}{4})$ following the merge. In the former case, only vertices in $P_2$ are recolored as long as capacity constraints are respected. In the latter case, the algorithm computes an optimal coloring $c_m$ for $P_1$. If $c_m(P_1)\neq c(P_1)$, where $c$ is the current coloring and recoloring $P_1$ according to $c_m$ respects capacity constraints, $P_1$ is recolored. If capacity constraints are violated at any point, \texttt{Follow-Greedy} calls (and transitions to) algorithm \texttt{Greedy-Recoloring($\mathcal{P}, W$)}. The complete pseudo-code of the algorithm is deferred to Appendix \ref{app: logncompetitive}.
\subparagraph{Analysis.} Whenever components are \textit{greedily} assigned an optimal coloring such that the total weight of components exceeds $W$, \texttt{Follow-Greedy} transitions to \texttt{Greedy-Recoloring}. At this point, we show that $OPT_{\sigma}=\Omega(W)$, and bound the recoloring cost incurred by \texttt{Greedy-Recoloring} on the request sequence thereafter by $O(W\log n)$. On the other hand, we show that the total recoloring cost incurred for any component $P$ whenever $w(P)$ increases (before \texttt{Greedy-Recoloring} is ever called) can be bounded by $O(\log n)$ times the recoloring cost incurred by OPT on $P$. Combining the two bounds yields $O(\log n)$-competitiveness. 

Lemma \ref{lemma: monotonicity} establishes a monotonicity property-- distance to optimal colorings for any component $P$ is monotonically non-decreasing with its weight $w(P)$.

\begin{restatable}{lemma}{monotonicity}
\label{lemma: monotonicity}
Let $P$ and $P'$ be components at times $t_1$ and $t_2$ respectively, where $1\leq t_1\leq t_2\leq T$ and $P\subseteq P'$. If $c_m$ and $c_m'$ are optimal colorings for $P$ and $P'$ respectively, then $d_{P}(c_m)\leq d_{P}(c_m')\leq d_{P'}(c_m')$.
\end{restatable}
\begin{proof}
We note that if $c_m(P)\neq c_m'(P)$, then 
\begin{align*}
d_P(c_m)=\sum\limits_{v\in P:\, c_m(v)\neq c_0(v)} w(v)\leq \sum\limits_{v\in P:\, c_m'(v)\neq c_0(v)} w(v)=d_P(c_m')
\end{align*}
since $c_m$ is an optimal coloring for vertices in $P$. Using this we have,
\begin{align*}
d_{P'}(c_m')&=\sum\limits_{v\in P:\, c_m'(v)\neq c_0(v)} w(v)+\sum\limits_{v\in P'\backslash P:\, c_m'(v)\neq c_0(v)} w(v) \geq \sum\limits_{v\in P:\, c_m'(v)\neq c_0(v)} w(v) \\
&=d_{P}(c_m'). \qedhere
\end{align*}\end{proof}
Our algorithm periodically recomputes an optimal coloring $c_m$ for any component when its weight increases by least a $(1+\frac{\varepsilon}{4})$ factor, and re-colors vertices according to $c_m$ (if needed). Let $T$ denote the time step before \texttt{Greedy-Recoloring} is called; if \texttt{Greedy-Recoloring} is never called, $T=|\sigma|$.
A component $P_i\in \mathcal{P}$ is said to be \textit{alive} during a time interval $T_{P_i}=[1,T']$ where $T'\leq T$, if for all requests $\sigma_t=(u_t,v_t)\in \sigma_{\leq T'}$ such that $u_t\in P_i, v_t\in P_j$, $w(P_i)\geq w(P_j)$; i.e. $P_i$ is not the \textit{smaller} of the two components merged during time $t\in T_{P_i}$.

\begin{restatable}{lemma}{deviationofcoloring}
\label{lemma: deviationofcoloring}
Let $P_i$ be an alive component during interval $T_{P_i}=[1,T']$, let $c_m^t$ denote an optimal coloring for $P_i$ at time $t$. For all $t\in T_{P_i}$ the coloring $c$ maintained by \texttt{Follow-Greedy} satisfies $d_{P_i}(c)\leq d_{P_i}(c_m^t)+\frac{\varepsilon}{4}w(P_i)$, where $w(P_i)$ is the weight of $P_i$ at time $t$.
\end{restatable}

\begin{proof}
Let $c(P_i)$ be the coloring maintained by the algorithm. Let $t_a, t_b\in T_{P_i}$ be time-steps such that $t_a<t_b$ at which the algorithm re-computes an optimal coloring for $P_i$ and recolors $P_i$ if necessary. Thus, $c(P_i)=c_m^{t}(P_i)$ at $t=t_a$ and $t=t_b$ respectively. For any $t\in (t_a, t_b)$, if $c=c_m^t$, then $d_{P_i}(c)=d_{P_i}(c_m^t)$. For all other $t$, $d_{P_i}(c_m^t)\geq d_{P_i}(c_m^{t_a})$ by Lemma \ref{lemma: monotonicity}. Let $S$ denote the set of vertices which are merged to $P_1$ in the time interval $(t_a+1, t_b)$. Since $w(S)\leq \frac{\varepsilon}{4}w(P_i)$, $d_{P_i}(c)\leq d_{P_i}(c_m^{t_a})+\frac{\varepsilon}{4}w(P_i)\leq d_{P_i}(c_m^t)+\frac{\varepsilon}{4}w(P_i)$, completing the proof.
\end{proof}

Lemma \ref{lemma: optlbW} gives a lower bound on $OPT_{\sigma}$ if \texttt{Greedy-Recoloring} is called.
\begin{restatable}{lemma}{optlbw}
\label{lemma: optlbW}
If \emph{\texttt{Follow-Greedy}} calls \emph{\texttt{Greedy-Recoloring}} at time $t$, then $OPT_{\sigma_{\leq t}}=\Omega(W)$.
\end{restatable}

\begin{proof}
Note that the residual capacity $N(C_i)=(1+\varepsilon)W-\sum_{v\in C_i} w(v)$ for any $i\in \{1,2\}$. Algorithm \texttt{Greedy-Recoloring} is called whenever following a merge of $P_2$ into $P_1$, either i) an optimal coloring for $P_1$ violates capacity constraints or ii) recoloring vertices in $P_2$ violates capacity constraints. Let $c$ denote the \textit{infeasible} coloring in any case, s.t. i) $c=c_m$ (the optimal coloring) when $w(P_1)>E(P_1)(1+\frac{\varepsilon}{4})$ or, ii) $c(v)=1$ for all $v\in B_2$, and $c(v)=2$ for all $v\in A_2$ (see lines 13-15, 20 of \texttt{Follow-Greedy}), such that $\sum_{P_i\in \mathcal{P}}\sum_{v\in P_i:\,c(v)=j} w(v)>W(1+\varepsilon)$ for some $j\in \{1,2\}$. Note that such a $j$ must exist since the capacity constraint is violated for one of the two clusters by the coloring $c$. By Lemma \ref{lemma: deviationofcoloring}, it follows that $d_{P_i}(c)\leq d_{P_i}(c_m)+\frac{\varepsilon}{4}w(P_i)$ where $c_m$ denotes an optimal coloring at time $t$ for any component $P_i$. Note that $OPT_{\sigma}\geq \sum_{P_i\in \mathcal{P}} d_{P_i}(c_m)$ since this lower bound disregards capacity constraints. Thus, $\sum_{P_i\in \mathcal{P}} d_{P_i}(c)\leq \sum_{P_i\in \mathcal{P}}[d_{P_i}(c_m)+\frac{\varepsilon}{4}w(P_i)]\leq OPT_{\sigma}+\frac{\varepsilon}{2}W$ since $\sum_{P_i\in \mathcal{P}}w(P_i)=2W$. Thus, $OPT_{\sigma}\geq \sum_{P_i\in \mathcal{P}} d_{P_i}(c)-\frac{\varepsilon}{2}W$. Now, note that $\sum_{P_i\in \mathcal{P}} d_{P_i}(c)\geq \varepsilon W$ since cluster capacities have been exceeded by at least $\varepsilon W$, which is only possible if at least a set $S$ of vertices with weight at least $\varepsilon W$ change their initial coloring. Thus, $OPT_{\sigma}\geq \varepsilon W-\frac{\varepsilon}{2}W=\frac{\varepsilon}{2}W=\Omega(W)$.  
\end{proof}

Lemma \ref{lemma:costofgreedy} bounds the recoloring cost of \texttt{Greedy-Recoloring}. The proof is similar to the proof of Theorem \ref{fdthm} and deferred to Appendix \ref{app: logncompetitive}.
\begin{restatable}{lemma}{costofgreedy}\label{lemma:costofgreedy}
The total recoloring cost incurred by \emph{\texttt{Greedy-Recoloring}} when it is called by \emph{\texttt{Follow-Greedy}} is $O(W\log n)$.
\end{restatable}

\begin{lemma} \label{main-cost-lemma}
\emph{\texttt{Follow-Greedy}} incurs a total recoloring cost of $O(OPT_{\sigma}\log n)$ before \emph{\texttt{Greedy-}} \emph{\texttt{Recoloring}} is ever called.
\end{lemma}
\begin{proof}
Let $T=|\sigma|$ if \texttt{Greedy-Recoloring} is never called; otherwise $T=t-1$ if it is called at time $t$. We analyze the cost of \texttt{Follow-Greedy} on the set of all \textit{alive} components $\mathcal{P}_A$ until time $T$. Any alive component $P'$ at any point satisfies $P'\subseteq P$ for some alive component $P\in \mathcal{P}_A$. Let $OPT(P)$ denote the cost incurred by $OPT$ on $P$ until time $T$. We show that the total cost of $\texttt{Follow-Greedy}$ on $P$ is bounded by $OPT(P)O(\log n)$ to conclude the proof. 

For any component $P\in \mathcal{P}_A$, let $I_{P}=\{[0=t_0, t_1), [t_1, t_2),[t_2, t_3),...,[t_f, t_{f+1}=T]\}$ denote the partition of time interval $[0,T]$ such that, for all $j\in [f]$ i) $P_{j}\subseteq P$ is a sub-component of $P$ for which an optimal recoloring is computed at time $t_j$, ii) $(1+\frac{\varepsilon}{4}) w(P_{j-1})\leq w(P_j)$ and iii) $P_{1}\subseteq P_2 \subseteq ... \subseteq P_{f+1}=P$. Recall that optimal colorings are only computed whenever component weights increase by at least a $(1+\frac{\varepsilon}{4})$ factor. Let $\mathcal{Q}_j$ denote the set of components that are merged to $P_{j-1}$ during time interval $[t_{j-1}, t_j]$ to yield $P_j$. Note that for components $Q\in \mathcal{Q}_j$, $w(Q)\leq w(P_{j-1})\leq (1+\frac{\varepsilon}{4})w(P_{j-1})$ by definition of $P_j$ and $\mathcal{Q}_j$ and the aforementioned property ii). Thus, $w(P_{j})\leq 2(1+\frac{\varepsilon}{4})w(P_{j-1})$. Moreover, by Lemma \ref{lemma: deviationofcoloring} the coloring $c$ maintained by the algorithm satisfies $d_{P_j}(c)\leq d_{P_j}(c_m^t)+\frac{\varepsilon}{4}w(P_j)$ where $c_m^t$ denotes an optimal coloring at any time $t\in [t_{j-1}, t_j)$ for $P_j$. There are two types of recolorings performed during time interval $[t_{j-1}, t_j]$:
\begin{enumerate}[-]
\item (\textit{Type I}) $P_j$ is recolored at time $t_j$.
\item (\textit{Type II}) $Q\in \mathcal{Q}_j$ is recolored at $t_Q\in [t_{j-1}, t_j]$ if $c(u_t)=c(v_t)$ where $u_t\in P_{j-1}$ and $v_t\in Q$. 
\end{enumerate} 
We charge the cost of both Type I and Type II colorings to $OPT$. For Type I colorings, note that $c_m^{t_{j-1}}(P_{j-1})\neq c_m^{t_j}(P_j)$, and hence $OPT(P)\geq OPT(P_j)\geq \frac{1}{2}w(P_{j-1})\geq \frac{1}{2}\frac{w(P_{j-1})}{2(1+\frac{\varepsilon}{4})}=\frac{w(P_j)}{4+\varepsilon}$ by Lemma \ref{lemma: monotonicity}. We charge every $v\in P_j$ s.t. $w(v)>\frac{w(P_j)}{2n}$ an amount $2w(v)$ to reflect this cost. Let $j_{max}\leq f$ denote the largest $j$ for which $P_j$ is recolored in this manner. Since these recolorings may happen every time a component undergoes a $(1+\frac{\varepsilon}{4})$ factor increase in weight, each vertex in $P_{j_{max}}$ is charged at most $O(\log_{(1+\frac{\varepsilon}{4})} n)=O(\log n)$ Type I recoloring costs. Accounting for all $P\in \mathcal{P}_A$, the total Type I recoloring cost is $O(OPT_{\sigma}\log n)$.

For Type II recolorings, consider a component $Q\in \mathcal{Q}_j$ that is merged to $P_{j-1}$, at time $t_Q$ for $j\in [f+1]$ and recolored (note that it is possible $Q$ is not recolored). If $P_{j}$ is recolored at time $t_j$ the Type II recoloring cost for $Q$ can be absorbed in the Type I recoloring cost for $P_j$ by charging each vertex $v\in P_j$ where $w(v)>\frac{w(P_i)}{2n}$ a cost $4w(v)$ instead of $2w(v)$ before. 

On the other hand, if $P_{j}$ is not recolored at time $t_j$, let $\mathcal{Q}_j'\subseteq \mathcal{Q}_j$ denote the set of smaller components that are recolored by \texttt{Follow-Greedy} when merged to $P_{j-1}$. By Lemma \ref{lemma: deviationofcoloring}, it follows that for all $t\in [t_{j-1}, t_Q)$, the coloring $c$ maintained by \texttt{Follow-Greedy} satisfies $d_{Q}(c)\leq d_Q(c_m^{t})+\frac{\varepsilon}{4}w(Q)$, where $c_m^{t}$ corresponds to an optimal coloring for $Q\in \mathcal{Q}_j'$ at time $t$. Moreover, $OPT(P_j)\geq d_{P_{j-1}}(c_m^{t_{j-1}})+\sum_{Q\in \mathcal{Q}_j'} \frac{1}{2}\frac{w(Q)}{(1+\frac{\varepsilon}{4})}=d_{P_{j-1}}(c_m^{t_{j-1}})+\sum_{Q\in \mathcal{Q}_j'} \frac{w(Q)}{(2+\frac{\varepsilon}{2})}$ since $P_j$ is not recolored. In other words, the optimal coloring for $P_j$ is sub-optimal for a set $S\subseteq Q$ of weight at least $\frac{w(Q)}{(1+\frac{\varepsilon}{4})}$ before time $t_Q$, when an optimal coloring is determined for component $Q$--and hence, incurs a cost of at least $\frac{1}{2}w(S)$ for recoloring $Q$. We charge every vertex $v\in Q$ for which $w(v)>\frac{w(Q)}{2n}$ a cost $2w(v)$. For any $j$ and $Q\in \mathcal{Q}_j$ for which Type II recolorings happen, every vertex $v$ in $Q$ can only be charged $O(\log n)$ times since the weight of the component containing $v$ doubles each time $v$ is charged. This yields a total cost of $O(OPT_{\sigma}\log n)$ for Type II recolorings, completing the proof.
\end{proof}

\begin{proof}[Proof of Theorem~\ref{thm: logncompetitive}]
We prove that \texttt{Follow-Greedy} is $O(\log n)$-competitive. If \texttt{Greedy-}\linebreak\texttt{Recoloring} is not invoked, Lemma \ref{main-cost-lemma} yields $O(\log n)$ competitiveness. Else, by Lemmas \ref{main-cost-lemma} and \ref{lemma:costofgreedy} a total cost of $O(OPT_{\sigma}{\log n}+W\log n)=O(W\log n)$ is incurred. Combined with Lemma \ref{lemma: optlbW}, this yields $O(\log n)$-competitiveness. 
\end{proof}

\section{Capacitated Online \texorpdfstring{\boldmath $\Delta$}{Delta}-recoloring}
\label{sec: Delta-coloring}

We give deterministic and randomized algorithms for capacitated $\Delta$-recoloring in an overprovisioned setting. Formally, we have $n$ vertices, $\Delta$ colors, each with capacity $(1+\varepsilon)\frac{n}{\Delta}$, such that the maximum degree of the graph induced by request sequence $\sigma$ is at most $(1-\varepsilon)\Delta$. We analyze competitiveness with respect to an optimal offline algorithm with \textit{no capacity constraints}. Initially all vertices are assigned to $\Delta$ colors such that each color has exactly $\frac{n}{\Delta}$ vertices (for convenience, we assume that $n$ is divisible by $\Delta$). 

Our algorithms maintain a list of feasible colors $L(v)\subseteq [1,\Delta]$ for each vertex $v\in V$. We derive a lower bound for $OPT_{\sigma}$ (similarly to \cite{azar2023competitive} for uncapacitated $(\Delta+1)$-recoloring) by considering the size of the minimum vertex cover of graph $G_M$ induced by the set of monochromatic edges $(u_t,v_t)$ (with respect to the initial coloring $c_0$), i.e. $c_0(u_t)=c_0(v_t)$. The following lemma is immediate since $OPT$ must recolor at least one endpoint of a monochromatic edge in $G_M$.
\begin{restatable}[\cite{azar2023competitive}]{lemma}{vertexCoverLowerBound}
\label{lemma: optdelta-lb}
$OPT_{\sigma}\geq |C^*|$, where $C^*$ denotes a minimum vertex cover of $G_M$.
\end{restatable}

Lemma~\ref{lemma: optdelta-lb} holds for the uncapacitated case \cite{azar2023competitive}; thus, our algorithms in the overprovisioned setting are competitive with respect to an optimal algorithm which is not capacity constrained.
\subparagraph{Maintaining a 2-approximate vertex cover.} Both of our algorithms maintain a $2$-approxi\-mate vertex cover $C$ of $G_M$ \textit{online}, i.e., $|C|\leq 2|C^*|$ using a simple greedy algorithm: on arrival of an edge $(u_t, v_t)$ such that $c_0(u_t)=c_0(v)$, if neither of $u_t$ or $v_t$ are in $C$, both $u_t$ and $v_t$ are added to $C$. Consider any distinct pairs of vertices $(u_i, v_i), (u_j,v_j)$ added to $C$. By virtue of the algorithm, vertices $u_i$ and $v_i$ can not be adjacent to $u_j$ or $v_j$ in $G_M$. Thus, any vertex cover must include one of $u_t$ or $v_t$ whenever $u_t, v_t$ are added by the greedy algorithm after a request $(u_t,v_t)$ is encountered. It follows that $C\leq 2|C^*|$.

\subparagraph{The Algorithm.} We give a generic algorithm \texttt{$\Delta$-Recoloring} below that invokes a \texttt{Recolor} subroutine. Our deterministic and randomized algorithms differ in their respective implementations of this subroutine. The vertex cover $C$ is initialized to $\emptyset$. On any request $(u_t,v_t)$, if $c(u_t)=c(v_t)$ there are multiple cases to consider: 
\begin{enumerate}[1.]
\item If $u_t, v_t\notin C$: add $u_t, v_t$ to $C$, and call \texttt{Recolor}($u_t$). 
\item Elsif $u_t\in C$, $v_t\notin C$: call \texttt{Recolor}($u_t$).
\item Else: if both $u_t, v_t\in C$, call \texttt{Recolor}($\arg\max_{w\in \{u_t,v_t\}}deg(w)$).
\end{enumerate}

\subsection{A Deterministic \texorpdfstring{\boldmath $O(\Delta)$}{O(Delta)}-Competitive Algorithm}
We give a deterministic algorithm by describing a \textit{deterministic} \texttt{Recolor} subroutine which is used in conjunction with algorithm \texttt{$\Delta$-Recoloring}.  As before, let $N(C_i)=(1+\varepsilon)\frac{n}{\Delta}-\sum_{v\in C_i} w(v)=(1+\varepsilon)\frac{n}{\Delta}-|C_i|$ denote the residual capacity of cluster $C_i$, for any $i$.

\subparagraph{Deterministic \texttt{Recolor}(\texorpdfstring{\boldmath $v$}{v}).}
Suppose $v$ is currently assigned to color $i$. Among the list of feasible colors $L(v)$ for $v$, we recolor $v$ with color $j$ with maximum residual capacity; i.e., $C_j=\arg\max_{\{C_i:\,i\in L(v), N(C_j)>0\}}N(C_i)$.  The variables $N(C_i)$ (resp. $N(C_j)$) are incremented (resp. decremented) by 1, and the list $L(v)$ updated. If $N(C_j)=0$ for all $j\in L(v)$, then the following \texttt{Rebalance} procedure is called and all colors are assigned at most $\lceil \frac{n}{\Delta}\rceil$ vertices.
\subparagraph{Deterministic \texttt{Rebalance}.} For any graph $G$ with maximum degree at most $r$, an equitable coloring is a proper coloring such that for any $i,j\in [r+1]$ the number of vertices assigned color $i$ differs by the number of vertices assigned color $j$ by at most 1. Kierstead et al. \cite{kierstead2010fast} present an $O(rn^2)$-time algorithm to compute a $(r+1)$-equitable coloring. We use their algorithm as the rebalancing procedure to recompute a coloring for the subgraph $G_t$ induced by $\sigma_{\leq t}$. The cost of \texttt{Rebalance} is trivially bounded by $O(n)$.

\begin{proof}[Proof of Theorem~\ref{thm:delta-1}]
Let $|C^*|$ denote the \textit{size} of the optimal vertex cover for the graph $G_M$. Our algorithm only recolors vertices in $C$. Since the maximum degree is $(1-\varepsilon)\Delta$, there are at least $\varepsilon \Delta$ feasible colors for each vertex when recoloring happens. If there exists a feasible color $j\in L(v)$ s.t. $N(C_j)\geq 1$, a cost of 1 is incurred. Thus, the total cost if \texttt{Rebalance} is never called throughout the request sequence is at most $O(2\cdot |C^*|\cdot(1-\varepsilon)\Delta)=O(|C^*|\Delta)$. Combined with Lemma \ref{lemma: optdelta-lb}, this yields $O(\Delta)$-competitiveness.

Now, let $v$ denote the first such vertex for which a rebalancing procedure is called. This implies that for at least $\varepsilon\Delta$ feasible colors $j\in L(v)$, s.t. $N(C_j)=0$. It follows that $OPT\geq |C^*|\geq \varepsilon \Delta(\frac{\varepsilon n}{\Delta}-1)=\varepsilon^2 n-\varepsilon\Delta\geq \frac{\varepsilon^2n}{2}$ in this case, since only vertices in $C$ are recolored. Since the sum of degrees of all nodes in the subgraph $G_T$ induced by $\sigma$ is $n\Delta(1-\varepsilon)$, the re-balancing procedure can be called $O(\frac{n\Delta(1-\varepsilon)}{\varepsilon^2n})=O(\Delta)$ times, incurring a total re-balancing cost $O(n\Delta)$ since a single call to re-balancing incurs $O(n)$ recolorings. The total number of recolorings between re-balancing calls is bounded by $O(n\Delta)$. Since $OPT=\Omega(n)$ in this case, the algorithm is $O(\Delta)$-competitive, completing the proof.
\end{proof}

\subsection{A Randomized \texorpdfstring{\boldmath $O(1)$}{O(1)} Competitive Algorithm}
\label{sec: randalg}
Our randomized algorithm uses \textit{randomized} \texttt{Recolor} and \texttt{Rebalance} subroutines in conjunction with algorithm \texttt{$\Delta$-Recoloring}. 
\subparagraph{Randomized \texttt{Recolor}\texorpdfstring{\boldmath $(v)$}{(v)}.}
Suppose $v$ is currently assigned color $i$.  We pick a feasible color $j$ for $v$ uniformly at random from $L(v)$. If $N(C_j)\geq 1$, recolor $v$ to $j$ and assign to $C_j$, and increment $N_C(i)$ (resp. decrement $N_C(j)$) by 1; otherwise, call randomized \texttt{Rebalance}. 
\subparagraph{Randomized \texttt{Rebalance}.}Initially, for each vertex $v$, the feasible list of colors is $L(v)=[\Delta]$. Thereafter, a coloring for $V$ is sequentially determined as follows: for each $v\in V$, a random color $c$ is picked from the current list of feasible colors $L(v)$, and assigned to $v$. For all neighbors $u$ of $v$, $c$ is removed from $L(u)$ and the process repeats.

The following holds since the maximum degree of any vertex is always bounded by $(1-\varepsilon)\Delta$.
\begin{observation}\label{obs: sizeoffeasiblelist} $|L(v)|\geq \varepsilon \Delta$ for any $v\in V$ and any time $t$.
\end{observation}

\subparagraph{Analysis of \texttt{Rebalance}.}  A key property of randomized rebalancing is that the distribution of the vertices across the colors is very close to balanced. This is formalized by the following lemma, whose proof utilizes a martingale argument and Azuma's inequality. 
\begin{restatable}{lemma}{rebalancinglemma}\label{lemma: rebalancinglemma}
Randomized \emph{\texttt{Rebalance}} computes a proper coloring $c$ satisfying i) the number of vertices assigned any color $i$ where $i\in [\Delta]$ is at most $(1+\frac{\varepsilon}{2})\frac{n}{\Delta}]$ with probability $1 - \Delta e^{-\frac{\varepsilon^2 n}{2\Delta^2}}$ and ii) the probability that any two vertices $u,v$ have the same color is at most $\frac{1}{\varepsilon\Delta}$. 
\end{restatable}
\begin{proof}
We first note that rebalancing always computes a proper coloring.  To establish a high probability upper bound on the number of vertices in any color, we use a martingale argument.  Let $X_1, X_2,...,X_n$ denote i.i.d.\ random variables uniformly generated from the interval $[0,1]$ in advance.  In the randomized \texttt{Rebalance} subroutine, we can view the color assignment for the $j$th vertex $v$ as follows: the algorithm selects the $k$th color in $L(v)$ for $k = \lceil X_j\cdot |L(v)| \rceil$. This ensures that any color $c\in L(v)$ is picked with probability $\frac{1}{|L(v)|}$. This view allows us to condition our martingale on independent random variables $X_1,..,X_i$ instead of the actual color choices, which are not independent in general.  

Let $Y_i$ denote the number of vertices assigned color $i$ once rebalancing is completed. Note that $Y_i$ depends on random variables $X_1,...,X_n$ all of which are mutually independent. Let $Z$ denote the Doob martingale where $Z_j=\mathbb{E}[Y|\,X_1,...,X_j]$, i.e. $Z_j$ is the expected value of $E[Y]$ given the random choices of the first $j$ vertices in $v$, and let $Z_0=E[Y]$. 

To compute a high probability bound on the value of $Z_n$, we first compute $\mathbb{E}[Z_0]=\E[Y_i]$. We note that the coloring computed by randomized \texttt{Rebalance} is symmetric with respect to any color. Thus, $\E[Y_i]=\E[Y_j]$ for any $i\neq j$ where $i,j\in [\Delta]$. Since $\E[\sum_{i\in [\Delta]}X_i]=n$, we have that $\E[Y_i]=\frac{n}{\Delta}$ for any $i\in \Delta$. Moreover, note that $Z_j-Z_{j-1}\leq 1$ for any $j\in [n]$.  We apply the Azuma-Hoeffding inequality (e.g., see Theorem 13.6 in \cite{mitzenmacher2017probability}) to derive 
\begin{align*}
\Pr\left[Z_n-Z_0\geq \frac{\varepsilon n}{2\Delta}\right]=\Pr\left[Z_n-\frac{n}{\Delta}\geq \frac{\varepsilon n}{2\Delta}\right]\leq e^{-\frac{2\varepsilon^2n^2}{4\Delta^2n}}=e^{-\frac{\varepsilon^2 n}{2\Delta^2}}
\end{align*}
Taking a union bound for all $\Delta$ colors, we have that the probability that for any color $i$, the number of vertices assigned color $i$ is at most $(1+\frac{\varepsilon}{2})\frac{n}{\Delta}$ is at least $1-\Delta e^{-\frac{\varepsilon^2 n}{2\Delta^2}}$. This completes the proof for property i) of coloring $c$ given by randomized rebalancing.

For property ii), let $v$ be colored after $u$ when randomized rebalancing is called. Note that $v$ has at least $\frac{1}{\varepsilon \Delta}$ feasible colors in $L(v)$ at the time it is colored (by Observation \ref{obs: sizeoffeasiblelist}) and chooses a color randomly. Since $u$ is a non-neighbor, $u$ can be assigned one of the colors in $L(v)$; since $v$ is assigned a random color in $L(v)$, we obtain $\Pr[c(u)=c(v)]\leq \frac{1}{\varepsilon\Delta}$.
\end{proof}

\subparagraph{Analysis of recoloring.} Our algorithm recolors only when $c(u_t)=c(v_t)$ for a request $(u_t,v_t)$ at any time $t$.  We first place an upper bound on the probability of recoloring. The proof of the following lemma is deferred to Appendix \ref{app: Delta-coloring}.
\begin{restatable}{lemma}{recoloringprobability}\label{lemma: recoloringprobability}
For a request $(u_t,v_t)$ in which at least one of $u_t$ or $v_t$ has been recolored before time $t$, $\Pr[c(u_t)=c(v_t)]\leq \frac{1}{\varepsilon \Delta}$.
\end{restatable}
For our analysis, we partition the request sequence into phases where a phase ends whenever a rebalancing occurs.  This yields a (possibly empty) sequence of complete phases followed by an incomplete phase.  Note that the length of any complete phase is a random variable.  We first show that each complete phase has $\Omega(n)$ recolorings with high probability.  
\begin{restatable}{lemma}{loadofcolors}\label{lemma: loadofcolors}
A complete phase has at least $\varepsilon^2 n/4$ recolorings with probability $1 - 2\Delta e^{-\frac{\varepsilon^2 n}{2\Delta^2}}$.
\end{restatable}
\begin{proof}
We establish the claim by showing that if at most $\frac{\varepsilon^2 n}{c}$ recolorings take place in phase, then 
the number of vertices recolored to any particular color $i$ is at most $(1+\varepsilon)\frac{n}{\Delta}$ with high probability, indicating that the phase has not completed. 
In the remainder of the proof, we assume that $t = \frac{\varepsilon^2 n}{4}$ recolorings take place.
Let $Y_i$ denote the random variable which is 1 if during the $i^{th}$ recoloring for some vertex $v$, $v$ was assigned a color $j$. Note that $\Pr[Y_i=1|\,Y_{i-1},...,Y_1]\leq \frac{1}{\varepsilon\Delta}$ by Observation \ref{obs: sizeoffeasiblelist}. We are interested in the random variable $Y=\sum_{i=1}^t Y_i$, and note that the total number of vertices assigned to color $i$ is bounded by $(1+\frac{\varepsilon}{2})\frac{n}{\Delta}$ at the beginning of a phase with high probability, by Lemma \ref{lemma: rebalancinglemma}. Moreover, $Y$ is stochastically dominated by $Z=\sum_{i=1}^t Z_i$, where $Z_1, Z_2,...,Z_t$ are i.i.d. Bernoulli random variables such that $Z_i=1$ with probability $\frac{1}{\varepsilon\Delta}$ and 0 otherwise. Note that $\E[Z]=\frac{\varepsilon n}{c\Delta}\leq \frac{\varepsilon n}{4\Delta}$. Thus, to upper bound $Y$, we can use a Chernoff bound (see p. 69 of \cite{mitzenmacher2017probability}) for $Z$:
\[
\Pr[Y\geq \frac{n}{4}\Delta]\leq \Pr[Z\geq \frac{\varepsilon n}{2}\Delta] \leq \Pr[Z\geq 2\E[Z]] \leq e^{-\frac{\varepsilon n}{12\Delta}}
\]
Taking a union bound over all $i$, we have that the probability the maximum number of vertices assigned any color $i\in \Delta$ is at most $(1+\varepsilon)\frac{n}{\Delta}$ with probability at least $1-\Delta e^{-\frac{\varepsilon n}{12\Delta}}$.  Adding the preceding failure probability to the failure probability of Lemma~\ref{lemma: rebalancinglemma} and noting that 
$\Delta e^{-\frac{\varepsilon n}{12\Delta}} \le 
\Delta e^{-\frac{\varepsilon^2 n}{2\Delta^2}}$ yields the desired claim.
\end{proof}

Lemma \ref{lemma: phaselength} bounds the number of new edges requested in any completed phase by $\Omega(n\Delta)$ with high probability. This bounds the number of phases by $O(1)$. 
\begin{restatable}{lemma}{phaselength}\label{lemma: phaselength} The number of new requests $(u_t,v_t)$ in any complete phase such that $(u_t, v_t)$ is requested for the first time at time $t$ is at least $\frac{\varepsilon^3 n\Delta}{4(1+\varepsilon)}$ with probability at least $1 - 3\Delta e^{-\frac{\varepsilon^2 n}{2\Delta^2}}$.
\end{restatable}
\begin{proof}
Let $Y_i$ denote the random variable which is $1$ if the $i^{th}$ newly requested edge $(u_i,v_i)$ during any phase $R$ leads to a recoloring by the algorithm, and 0 otherwise. By Lemma \ref{lemma: recoloringprobability}, $\Pr[Y_i=1| Y_{i-1},...,Y_1]\leq \frac{1}{\varepsilon\Delta}$. We use the fact that $Y=\sum_{i=1}^{N} Y_i$ is stochastically dominated by $Z=\sum_{i=1}^N D_i$ where $D_i$ is a Bernoulli random variable that is 1 with probability $\frac{1}{\varepsilon\Delta}$ and $0$ otherwise. Thus, we can apply a Chernoff bound (see p. 69 of \cite{mitzenmacher2017probability}) to $Z$ to get a high probability bound on $Y$. Let $N=\frac{\varepsilon^3n\Delta}{4(1+\varepsilon)}$. We analyze the probability that $Z=\sum_{i=1}^N D_i\geq \frac{\varepsilon^2 n}{4}$, i.e. the probability that the number of recolorings is at least $\frac{\varepsilon^2 n}{4}$ recolorings for $N$ new requests in the phase. Then $\E[Z]=\frac{\varepsilon^3 n\Delta}{4\varepsilon (1+\varepsilon)\Delta}=\frac{\varepsilon^2 n}{4(1+\varepsilon)}$. We derive
\begin{align*}
\Pr\left[Y\geq \frac{\varepsilon^2 n}{4}\right]&\leq \Pr\left[Z\geq \frac{\varepsilon^2 n}{4}\right]=\Pr\left[Z\geq (1+\varepsilon)E[Z]\right] \leq e^{-\frac{\varepsilon^4n}{12(1+\varepsilon)}}.
\end{align*}  
We obtain the desired claim by adding the above failure probability to that of  Lemma~\ref{lemma: loadofcolors}.
\end{proof}

Lemma \ref{lemma: firstphaserecolorings} bounds the expected number of recolorings throughout the algorithm.
\begin{restatable}{lemma}{firstphaserecolorings}\label{lemma: firstphaserecolorings}
The expected number of recolorings performed during the first phase is  $O(|C^*|)$.  The expected number of recolorings of the algorithm is $O(n)$.
\end{restatable}

\begin{proof}
From Lemma \ref{lemma: recoloringprobability}, it follows that at any time $t$ and any vertex $v$ such that $u$ is not a neighbor of $v$ at time $t$, $\Pr[c(u)=c(v)]\leq \frac{\varepsilon}{\Delta}$. Since only vertices in the vertex cover $\mathcal{C}$ are recolored, the number of events for which a recoloring may happen is $|\mathcal{C}|\Delta(1-\varepsilon)$. Each such event happens with probability at most $\frac{1}{\varepsilon \Delta}$ and since $|\mathcal{C}|\leq 2|C^*|$, it follows that the expected number of recolorings by the algorithm is bounded by $2|C^*|\frac{(1-\varepsilon)}{\varepsilon}=O(|C^*|)$.

To establish the second statement of the lemma, 
we use Lemma~\ref{lemma: recoloringprobability} and the fact that the total number of distinct edges is bounded by $n\Delta$ to derive that   
the expected total number of recolorings of the algorithm is at most $n + n\Delta/(\varepsilon \Delta) = n(1 + 1/\varepsilon)$.
\end{proof}
Finally, combining the lemmas above yields $O(1)$-competitiveness for our randomized algorithm, thus proving Theorem~\ref{thm:delta-2}.

\newcommand{\poly}{\mbox{poly}}

\begin{proof}[Proof of Theorem~\ref{thm:delta-2}]
By Lemma~\ref{lemma: optdelta-lb}, $OPT_{\sigma} \ge |C^*|$.  We consider two cases.  First, if $|C^*| \le \frac{\varepsilon^2 n}{4(1+\varepsilon)}$, then the number of edges introduced is at most $\frac{\varepsilon^3 n\Delta}{4(1+\varepsilon)}$.  By Lemma~\ref{lemma: phaselength}, with probability at least $1 - 3\Delta e^{-\frac{\varepsilon^2 n}{2\Delta^2}}$, the first phase does not complete.  The failure probability is $1/\poly(n)$ for $\Delta = O(\sqrt{n/\log n})$. Therefore, by Lemma~\ref{lemma: firstphaserecolorings}, the expected cost of the algorithm is $O(|C^*|) + n\Delta/\poly(n)  = O(|C^*|)$.

Second, if $|C^*| > \frac{\varepsilon^2 n}{4(1+\varepsilon)}$, then by Lemma~\ref{lemma: firstphaserecolorings}, the expected cost of all the recolorings is $O(n)$.  By Lemma~\ref{lemma: phaselength}, the number of phases is $O(1)$ with probability at least $1 - 3\Delta e^{-\frac{\varepsilon^2 n}{2\Delta^2}}$.  Therefore, the expected rebalancing cost is at most $(1 - 1/\poly(n))O(n) + n\Delta/\poly(n) = O(n)$.  Therefore, the expected total cost is $O(|C^*|)$, thus completing the proof.
\end{proof}

\bibliography{p096-Rajaraman}
\appendix
\section{Lower bounds} 
\label{app: lowerbounds}
In the following sections, we present lower bounds for capacitated vertex recoloring in the standard and fully dynamic variants of the problem for $k=2$. For fully dynamic recoloring, we give an $\Omega(n)$ deterministic lower bound for an arbitrary amount of resource augmentation $\varepsilon>0$. For the standard variant, we give an $\Omega(\log n)$ lower bound that holds for both deterministic and randomized algorithms. 

\subsection{Proof of Theorem \ref{lowerbdet}}
\label{app: lbdet}
\lowerbdet*
\noindent \begin{proof}
    We construct an appropriate request sequence $\sigma$ and argue that there exists an offline algorithm for which the competitive ratio achieved by any deterministic algorithm on that sequence is $\Omega (n)$. Let $w(v)=1$ for all $v\in V$. The request sequence consists of edges on an odd cycle of length $\ell=\Omega(n)$ where $\ell=\lfloor \frac{n}{2}\rfloor$ if $\lfloor \frac{n}{2}\rfloor$ is odd, else $\ell=\lfloor \frac{n}{2}\rfloor-1$. Fix an arbitrary set $S$ of $\ell$ vertices $v_1, v_2,...,v_{\ell}$. The request sequence will consist of edges from the cycle $C_S=((v_1, v_2), (v_2, v_3),...,(v_{\ell-1}, v_{\ell}), (v_{\ell}, v_1))$. Now consider $\ell$ offline algorithms $OFF_1, OFF_2,..,OFF_{\ell}$ where for all $i\in [\ell]$, $OFF_i$ maintains a coloring of vertices such that edge $(v_i, v_{(i+1)\text{mod } \ell})$ is the only monochromatic edge, i.e. $c(v_i)= c(v_{(i+1)\text{mod } \ell})$. 
    
    It is easy to observe that for any request $(v_i, v_{(i+1)\text{mod } \ell})$ on the cycle, exactly one offline algorithm $OFF_i$ incurs a cost. In that case, $OFF_i$ recolors $v_i$, and immediately recolors $v_i$ after the request incurring a total cost of $2$, ensuring that the monochromatic edge $(v_i, v_{(i+1)\text{mod } \ell})$ on $C_S$ prior to the request is unchanged. The (total) initial cost incurred by the $\ell$ offline algorithms to color nodes on the cycle such that exactly one edge in $C_S$ is monochromatic is bounded above by $n^2$. From the above, it follows the total cost incurred by the $\ell$ offline algorithms to serve the sequence is at most $\ell^2+2|\sigma|$. By an averaging argument it follows that there exists $j\in [\ell]$ such that $OFF_j$ incurs cost at most $\ell+\frac{2|\sigma|}{\ell}$.

    We now describe how to adaptively construct $\sigma$ such that the cost incurred by any online algorithm $\mathcal{A}$ is at least $|\sigma|$; precisely, all requests correspond to a monochromatic edge that $\mathcal{A}$ has on the cycle $C_S$. Since $C_S$ is an odd cycle, $\mathcal{A}$ must have at least one monochromatic edge on $C_S$. Hence, the total cost incurred by $\mathcal{A}$ is lower bounded by $|\sigma|$. Thus, for $|\sigma|\geq \ell^2$, the cost incurred by an offline algorithm $OFF_j$ is $O(\frac{A_{\sigma}}{\ell})$, implying that the competitive ratio of $\mathcal{A}$ is $\Omega(\ell)=\Omega(n)$, completing the proof.
\end{proof}

\subsection{Proof of Theorem \ref{lowerbrand}}
\label{app: randlognlb}
\lowerbrand*
The proof is similar to Theorem 4 in \cite{rajaraman2022improved}, and Theorem 5.6 in \cite{azar2023competitive}.

\noindent\begin{proof} 
We assume $n=2k$ is a power of 2 (i.e. for the sake of the proof (otherwise, set $n=2^{\lfloor \log (2k) \rfloor}$). We first consider the deterministic lower bound. The request sequence $\sigma$ is presented in successive batches $B_1,...,B_{\log n}$, where in each batch $B_i$, there are exactly $2^{n-i}$ requests such that each request is between endpoints of two distinct components of size exactly $2^{i-1}$ with the same color. More precisely, let $C_{i-1}=\{P_1,P_2,...,P_{\frac{n}{2^{i-1}}}\}$ be the bipartite components (each of which is a simple path of even length) before batch $B_i$ is presented. Every component $P_a$ is uniquely paired with another component $P_b$, for a total of $\frac{n}{2^{i-2}}$ pairs of components. Each request in batch $B_i$ is a request between two endpoints of a pair $(P_a, P_b)$ that have the same color in a coloring maintained by an online algorithm $\mathcal{A}$. As a result, for each pair $(P_a, P_b)$, $\mathcal{A}$ incurs cost at least $2^{i-1}$, for a total cost of $\frac{n}{2^{i-2}}\cdot 2^{i-1}=\frac{n}{2}$. This amounts to a total cost of $\Omega(n\log n)$. The optimal offline algorithm pays cost at most $n$ to compute an optimal coloring. Thus, the competitive ratio of any deterministic online algorithm is $\Omega (\log n)$. 

For the randomized lower bound, we invoke Yao's minimax Lemma and construct a distribution of requests such that any deterministic algorithm incurs expected cost $\Omega(n\log n)$ on the entire sequence. The distribution is fairly simple: for batch $B_i$, a request $(u,v)$ is presented where $u$ and $v$ are randomly chosen endpoints of components $P_a$ and $P_b$ respectively. In this case, any deterministic algorithm must incur an expected cost of at least $\frac{1}{2^{i-1}}$. Thus, the preceding argument yields a lower bound of $\Omega(\log n)$ for randomized algorithms. 

Finally, note that by virtue of the request sequence, any algorithm is forced to keep exactly $\frac{n}{2}$ nodes of each color at any time. Hence, any amount of augmentation is futile.
\end{proof}

\subsection{Proof of Theorem \ref{lowerbdelta}}
\label{app: lbdelta}
\lowerbdelta*
\begin{proof}
The idea of the proof is similar to Theorem 4.8 of \cite{azar2023competitive}, The adversary maintains a set $S$ of $\Delta+1$ vertices at any point. Since the number of available colors are $\Delta$, there always exist two vertices $u,v\in S$, such that $c(u)= c(v)$ where $c$ is the coloring maintained by a deterministic online algorithm. Without loss of generality, we assume that $n$ is divisible by $\Delta$.

Starting with an arbitrary set $S$ of $\Delta+1$ vertices, the adversary does the following. It repeatedly adds an edge $(u,v)$ between vertices $u,v\in S$, s.t. $c(u)=c(v)$, until the degree of at least one vertex $w$ in $S$ becomes $\Delta(1-\varepsilon)$. Then it removes all vertices from $S$ with degree $\Delta(1-\varepsilon)$ (note that there are at most 2 such vertices). Let $r$ denote the number of removed vertices. It then adds $r$ isolated vertices to $S$ and repeats the process. The number of times such a process can be repeated is $\Theta(n)$ since at most 2 vertices leave $S$. 

Note that after $i$ such processes, the number of edges added are at least $\frac{i\Delta}{2}$, following from a lower bound on the sum of degrees of vertices removed from $S$ after $i$ iterations. For each edge in the request sequence, the deterministic algorithm incurs a recoloring cost of 1. Consider the case when the number of processes is $\Omega(n)$: the total cost paid by the online algorithm is $\Omega(n\Delta)$. On the other hand, for any graph with maximum degree upper bounded by $\Delta$, there exists an equitable coloring such that the number of vertices assigned any color is $\frac{n}{\Delta}\pm 1$ \cite{kierstead2010fast}. The optimal offline algorithm simply incurs an initial cost of $O(n)$ to compute an equitable coloring for the graph induced by the request sequence. Thus, any deterministic algorithm for $(1+\varepsilon)$-overprovisioned $\Delta$-recoloring problem has a competitive ratio of $\Omega(\Delta)$. 
\end{proof}
\newpage

\section{A Rebalancing Procedure}\label{sec: Apprebalance}

We give details of the re-balancing subroutine $\texttt{Rebalance}$ which takes as input a weight parameter $W'$, component set $\mathcal{P}$, and $\varepsilon$. It computes a an assignment of components, denoted as tuples, where for each tuple $(i,j)$, $i$ denotes the index of a component $P_i\in \mathcal{P}$ and $j\in \{1,2\}$ denotes $A_i$ is assigned on $C_j$, and $B_i$ on $C_{3-j}$, such that the total weight of components assigned on $C_1$ is at most $W'(1+\varepsilon')$. If such an assignment is not possible, it returns infeasible. We give a FPTAS (fully polynomial-time approximation scheme) which takes $O(\frac{n^2 \ln W'}{\varepsilon})$ time. 

\begin{algorithm}[ht]
\caption{\texttt{Rebalance} ($\mathcal{P}, W', \varepsilon$)}
    \begin{algorithmic}[1]
        \State $T \leftarrow \emptyset$ \Comment{$T$ is a 2D-list where $T[i]$ is a list of feasible assignments of components $P_1,P_2,..,P_i$}.
        \State $T[1]=\{(w(A_1), \{(1,1)\}), (w(B_1), \{(1,2)\})\}$ \Comment{An element of $T[i])$ is a tuple $(x,y)$ s.t. $x$ is the total weight of the assignment specified by $y$.}
        \For{$i=2$ to $n$}
            \State $X\leftarrow T[i-1]\cup \{(x+|A_i|, y\cup\{(i,1)\})|\, (x,y)\in T[i-1]\}$.
            \State $Y\leftarrow T[i-1]\cup \{(x+|B_i|, y\cup\{(i,2)\})|\, (x,y)\in T[i-1]\}$.
            \State $T[i]\leftarrow \texttt{Merge}(T[i],X, Y)$. \Comment{\texttt{Merge} merges lists in sorted order w.r.t. total assignment weights.}
            \State $(a,b)\leftarrow T[i][-1]$. \Comment{$(a,b)$ is the last tuple in $T[i]$}
            \For {$j=(|T[i]|-1)$ to $1$}
                \State $(x,y)\leftarrow T[i][j]$.
                \If{$a\geq \frac{x}{1+\frac{\varepsilon}{2n}}$}
                    \State Mark $(x,y)$. 
                \Else 
                    \State $(a,b)\leftarrow (x,y)$.
                \EndIf 
            \EndFor 
            \State Remove all marked tuples from $T[i]$.
            \For {$(x,y)\in T[i]$}
                \If {$x> (1+\varepsilon)W'$}
                    \State $T[i]\leftarrow T[i]\backslash \{(x,y)\}$.
                \EndIf 
            \EndFor 
    \EndFor
    \If {$T[n]\neq \emptyset$} 
        \State Find a tuple $(x,y)\in T[n]$ such that $x\geq W'$. 
        \For{all $(i,j)\in y$}
            \State For all $v\in A_i$, call \texttt{Recolor$(v, j)$}.
            \State For all $v\in B_i$, call \texttt{Recolor$(v,3-j)$}.
        \EndFor 
   \Else 
   \State Return \textbf{infeasible}.
    \EndIf
    \end{algorithmic}
\end{algorithm}

We describe the \texttt{Rebalance} procedure. It maintains an array $T$ such that $T[i]$ corresponds to a list of assignments for $P_1,P_2,...,P_i\in \mathcal{P}$. Each assignment is a tuple $t=(x,y)$ for which $x$ denotes the total weight of components assigned to $C_1$, and $y$. Each entry $(i,j)\in y$ corresponds to an assignment of component $P_i=(A_i, B_i)$ where $A_i$ is assigned on $C_j$ (and $B_i$ on $C_{3-j}$). For the base case, w.l.o.g. assume $w(A_1)\leq w(B_1)$. Then, $T[1]$ consists of two possible assignments corresponding to $A_1$ (resp. $B_1$) being assigned to $C_1$ (resp. $C_2$) and  $B_1$ (resp. $A_1$) being assigned to $C_1$ (resp. $C_2$). $T[i]$ is obtained by considering various possibilities that $A_i$ and $B_i$ can be assigned given assignments in $T[i-1]$. The \texttt{Merge} subroutine takes 3 lists as arguments, and merges them in a way such that the resulting list, $T[i]$ is sorted w.r.t non-decreasing order of the first entry of the tuple (corresponding to the weight of an assignment).

In lines 9-14, we prune $T[i]$ by removing adjacent tuples (assignments) for which there is another assignment within factor $(1+\frac{\varepsilon}{2n})$ of its weight. This step effectively ensures that the size of the list $T[i]$ is bounded by $O(n\log W')$ for all $i$ (polynomial in the input representation of $W'$). In line 17, we remove all assignments from $T[i]$ which have weight more than $W'(1+\varepsilon)$. Finally, components are recolored if a feasible assignment exists.

\begin{lemma}
    \texttt{Rebalance} runs in time $O(\frac{n^2 \ln W'}{\varepsilon})$. Moreover, if there exists an assignment of components with total weight $W'$ on $C_1$, then \texttt{Rebalance} recolors vertices such that the total weight of vertices assigned on $C_1$ is at most $(1+\varepsilon) W'$, 
\end{lemma}
\begin{proof}
    We prove that the length of $T[i]$ is bounded by $O(\frac{n \ln W'}{\varepsilon})$ for all $i$. The claim clearly holds for $i=1$. Note that after the pruning step in 9-14, successive tuples $t_1, t_2$ in $T[i]$ differ in their first coordinate by at least $(1+\frac{\varepsilon}{2n})$. As a result, the length of $T[i]$ is bounded by $O(\log_{(1+\frac{\varepsilon}{2n})}W')=O(\frac{\ln W'}{\ln (1+\frac{\varepsilon}{2n})})=O(\frac{4n\ln W'}{\varepsilon})$ where the last inequality follows from the fact that for $z\in (0,1)$, $\ln(1+z)\geq \frac{z}{2}$. Moreover, the procedure \texttt{Merge} runs in time $|T[i]|+|X|+|Y|=O(|T[i]|)$. Thus, the total running time taken by $\texttt{Rebalance}$ is $O(\frac{n^2\ln W}{\varepsilon})$. 

Note that since all tuples with weight at least $(1+\varepsilon)W'$ are removed from $T[i]$, for all $i$, $(x,y)$ satisfies that $x<(1+\varepsilon)W'$. For any feasible assignment of subset of components in $\{P_1, P_2,...,P_i\}$, such that the total weight of vertices assigned to $C_1$ is $B$, there exists a tuple $(x,y)\in T[i]$ such that $x\leq (1+\frac{\varepsilon}{2n})^iB \leq e^{\frac{i\varepsilon}{2n}}B$ since the trimming procedure only marks tuples for which there exists a another tuple with weight at most $(1+\frac{\varepsilon}{2n})$. This implies that if there exists an assignment of weight $W'$, there exists $(x,y)\in T[n]$ such that $x\leq (1+\frac{\varepsilon}{2n})^n W'\leq e^{\frac{\varepsilon}{2}}W'\leq (1+\varepsilon)W'$ where the last inequality follows from the fact that for any $z\in (0,1)$, $e^{\frac{z}{2}}\leq 1+z$.
\end{proof}
\newpage
\section{Pseudo code of \texttt{Greedy-Recoloring}}
\label{app: fullydynamic} 
We give the pseudo code of our algorithm \texttt{Greedy-Recoloring} as follows. 
\begin{algorithm}[ht]
    \caption{\texttt{Greedy-Recoloring}($\mathcal{P}, W$)}
    \begin{algorithmic}[1]
        \For {every request $(u_t,v_t)$ revealed online}  
            \If{$u_t, v_t\in P_1=(A_1,B_1)$}
                \If {$u_t, v_t\in A_1$ or $u_t,v_t\in B_1$} \Comment{$P_1$ ceases to be bipartite, and a new phase begins.}
                    \State \texttt{Recolor}($u_t, 3-c(u_t)$). 
                    \State \texttt{Rebalance}$(\{\{v\}|\,v\in V\},W, \frac{\varepsilon}{2})$. 
                    \State Start a new phase by calling \texttt{Greedy-Recoloring($\mathcal{P}, W$)}.
                \EndIf 
            \Else \Comment{$u_t\in P_1, v_t\in P_2$ s.t. $P_1\neq P_2$ and wlog. $w(P_1)\geq w(P_2), s.t. A_1\cup A_2\subseteq C_1, B_1\cup B_2\subseteq C_2$.}
                \State $w(P_1)\leftarrow w(P_1)+w(P_2)$. \Comment{$P_2$ is merged to $P_1$.}
                \If {$u_t\in A_1$ and $v_t\in B_2$ or $u_t\in A_2$ and $v_t\in B_1$}
                    \State $P_1\leftarrow (A_1\cup A_2, B_1\cup B_2)$. 
                 \Else \Comment{$u_t\in A_1$ and $v_t\in A_2$ or $u_1\in B_1, v_t\in B_2$}
                    \State $P_1\leftarrow (A_1\cup B_2, B_1\cup A_2)$. 
                    \If{$N(C_1)-w(A_1)\geq w(B_2)$ and $N(C_2)-w(B_2)\geq w(A_2)$ and $w(P_2)\leq \frac{\varepsilon W}{4}$}
                        \State For all vertices $v\in A_2$, \texttt{Recolor$(v,2)$}.
                        \State For all vertices $v\in B_2$, \texttt{Recolor$(v,1)$}.
                        \State $\mathcal{P}\leftarrow \mathcal{P}\backslash \{P_2\}$.
                    \Else 
                       \State $\mathcal{P}\leftarrow \mathcal{P}\backslash \{P_2\}$.
                       \If{$\texttt{Rebalance}(\mathcal{P},W, \frac{\varepsilon}{2})$ returns \textbf{infeasible}}
                            \State \texttt{Recolor}($u_t, 3-c(u_t)$). 
                            \State \texttt{Rebalance}$(\{\{v\}|\,v\in V\},W, \frac{\varepsilon}{2})$. 
                            \State Start a new phase by calling \texttt{Greedy-Recoloring($\mathcal{P}, W$)}.
                        \EndIf 
                    \EndIf 
                \EndIf 
            \EndIf
        \EndFor
    \end{algorithmic}
\end{algorithm}
\newpage
\section{Pseudo code and Proofs from Section \ref{sec: logncompetitive}}
\label{app: logncompetitive}
\subsection{Pseudo code of \texttt{Follow-Greedy}}
We give the pseudo code of our algorithm \texttt{Follow-Greedy} as follows. 

\begin{algorithm}[ht]
\caption{\texttt{Follow-Greedy}}
\begin{algorithmic}[1]
\State $\mathcal{P}\leftarrow \emptyset$.
    \For {all $v\in V$}
        \State $P\leftarrow \{v\}$ s.t. $E(P)\leftarrow w(v)$.
        \State $\mathcal{P}\leftarrow \mathcal{P}\cup \{P\}$. 
    \EndFor
        \For {every request $(u_t, v_t)$ revealed online} \Comment{w.l.o.g. $u_t\in P_1=(A_1, B_1), v_t\in P_2=(A_2, B_2)$ s.t. $w(P_1)\geq w(P_2)$ s.t. $A_1\cup A_2\subseteq C_1, B_1\cup B_2\subseteq C_2$}
            \State $w(P_1)\leftarrow w(P_1)+w(P_2)$.
            \State $\mathcal{P}\leftarrow \mathcal{P} \backslash \{P_2\}$.  \Comment{$P_2$ is merged to $P_1$.}
            \If {$u_t\in A_1$ and $v_t\in B_2$ or $u_t\in A_2$ and $v_t\in B_1$}
                \State $P_1\leftarrow (A_1\cup A_2, B_1\cup B_2)$. 
            \Else \Comment{$u_t\in A_1$ and $v_t\in A_2$ or $u_t\in B_1, v_t\in B_2$}
                \State $P_1\leftarrow (A_1\cup B_2, B_1\cup A_2)$.
                \If {$w(P_1)<E(P)(1+\frac{\varepsilon}{4})$}
                    \If {$N(C_1)+w(A_2)\geq B_2$ and $N(C_2)+w(B_2)\geq B_1$}
                        \State For all vertices $v\in B_2$, call \texttt{Recolor($v, 1$)}.
                        \State For all vertices $v\in A_2$, call \texttt{Recolor($v, 2$)}.
                    \Else 
                        \State Call \texttt{Greedy-Recoloring($\mathcal{P}, W$)}.     
                    \EndIf 
                \EndIf 
            \EndIf 
            \If {$w(P_1)>E(P_1)(1+\frac{\varepsilon}{4})$}
                \State $E(P_1)\leftarrow w(P_1)$.
                \State Compute an optimal coloring $c_m$ for $P_1$.
                \If{$N(C_1)+w(A_1\cup A_2)\geq \sum\limits_{v\in V:\, c_m(v)=1} w(v)$ and $N(C_2)+w(B_1\cup B_2)\geq \sum\limits_{v\in V:\, c_m(v)=2} w(v)$} 
                    \State For all vertices $v\in P_1$, \texttt{Recolor($v, c_m(v)$)}.
                \Else
                    \State Call \texttt{Greedy-Recoloring$(\mathcal{P},W)$}.
                \EndIf 
            \EndIf             
        \EndFor
\end{algorithmic}
\end{algorithm}

\subsection{Proof of Lemma \ref{lemma:costofgreedy}}
\costofgreedy*
\begin{proof}
    Since we are guaranteed the existence of a proper coloring which respects capacities, lines 2-6 and 20-22 are never executed when \texttt{Greedy-Recoloring} is called. For the sake of analysis, we say a component $P_i\in \mathcal{P}$ is \textit{light} if $w(P_i)\leq \frac{\varepsilon}{4}W$ and \textit{heavy} otherwise. Whenever a component merge happens, the algorithm always recolors vertices in the lighter component $P_2$ as long as $P_2$ stays small and capacity constraints are not violated. In this case, $w(P_1)\geq 2w(P_2)$, and we charge each vertex $v$ in component $P_2$ where $w(v)>\frac{w(P_2)}{2n}$ a cost $2w(v)$. Since any component has $n$ vertices, $P_2$ must contain always contain a subset of vertices whose total weight is at least $\frac{w(P_2)}{2}$. We call this the \textit{Type 1} charge for a vertex $v$.
    
    On the other hand, whenever $P_2$ is small and vertices in $P_2$ cannot be recolored due to capacity violations, \texttt{Rebalance} is called. In this case, since $P_2$ is small, it follows that a total weight of at least $(1+\varepsilon)W-(1+\frac{\varepsilon}{2})W-w(P_2)\geq \frac{\varepsilon W}{4}$ vertices must have been successfully recolored since the last call to \texttt{Rebalance}. Let us call the set of successfully recolored vertices $S$, s.t. $w(S)\geq \frac{\varepsilon W}{4}$. Calling $\texttt{Rebalance}$ incurs a total recoloring cost of at most $2W$. We charge this cost to vertices in $S$ such that every $v\in S$ s.t. $w(v)\geq \frac{w(S)}{2n}$ is charged a cost of $\frac{16}{\varepsilon}w(v)$. We call this the \textit{Type 2} charge for a vertex $v$.

    If $P_2$ is heavy, \texttt{Rebalance} is always called. The total cost of at most $2W$ is charged to vertices in $P_2$, such that each vertex $v$ s.t. $w(v)>\frac{w(P_2)}{2n}$ is charged a cost of $\frac{16}{\varepsilon}w(v)$. We call this the \textit{Type 3} charge for $v$.

    We show that the total charge accumulated at every vertex $v$ is $O(w(v)\log n)$. This completes the proof since $\sum_{v\in V}O(w(v))\log n)=O(W\log n)$. Observe that after a Type 1 or Type 3 charge to $v$, the weight of its component doubles. It follows that any vertex $v$ where $w(v)>\frac{w(P_2)}{2n}$ can only be charged $O(\log n)$ times a cost of $O(w(v))$ before it is part of a component $P_2$ such that $w(v)\leq \frac{w(P_2)}{2n}$. Finally, note that whenever vertex $v\in S$ receives a Type 2 charge, the size of its component also doubles since the last call to \texttt{Rebalance}. Summarizing, any vertex $v\in V$ is charged at most $O(\frac{1}{\varepsilon})$ times a Type 3 charge and $O(\log n)$ times a Type 1 or Type 2 charge, giving a total charge of $O(\log n)w(v)$ for $v$.
\end{proof}

\newpage
\section{Missing Proofs from Section~\ref{sec: Delta-coloring}}
\label{app: Delta-coloring}
\recoloringprobability*
\begin{proof}
    If the condition of the lemma holds and this is the first step in the current phase when either $u_t$ or $v_t$ is being recolored, then the claim follows from property~(ii) of Lemma~\ref{lemma: rebalancinglemma}.  Otherwise, let $v_t$ denote the last vertex that was recolored (w.l.o.g.). Note that if $v_t$ is a neighbor of $u_t$ before time $t$, then $\Pr[c(u_t)=c(v_t)]=0$. We bound the probability in the case when $v_t$ is not a neighbor of $u_t$ before time $t$. Let $t'$ denote the last time $v_t$ was recolored, and let $d_{v_t}^{t'}$ denote the degree of $v_t$ at time $t'$. Then, $v_t$ had at least $\Delta-d_{v_t}^{t'}$ colors in $L(v_t)$ at time $t'$ to choose from. Consider any set $A'\subseteq [\Delta]$ of colors at time $t$ where $c:\, V\backslash\{v_t\}\rightarrow A'$ and $S'$ denote the set of all feasible colors for $v_t$ at time $t'$ given $A'$. Let $\mathcal{A}$ denote the set of all possible subsets $A'$. We have that, 
    \begin{align*}
        \Pr[c(u_t)=c(v_t)]&=\sum_{A'\in \mathcal{A}}\sum_{i\in S'}\Pr[L(v_t)=S']\cdot \Pr[c(u_t)=i|\, L(v_t)=S']\cdot \Pr[c(v_t)=i|\, c(u_t)=i, L(v_t)=S'] \\
        &=\sum_{A'\in \mathcal{A}} \Pr[L(v_t)=S']\frac{1}{|S'|}\sum_{i\in S'}\Pr[c(u_t)=i|\, L(v_t)=S'] \\
        &\leq \sum_{A'\in \mathcal{A}}\Pr [L(v_t)=S']\frac{1}{|S'|}\cdot 1 \leq \sum_{A'\in \mathcal{A}}\Pr [L(v_t)=S']\frac{1}{\Delta-d_{v_t}^{t'}} \leq \frac{1}{\varepsilon\Delta}
    \end{align*}
\end{proof}

\end{document}